\RequirePackage{silence}
\documentclass[a4paper,onecolumn,accepted=2026-09-18]{quantumarticle}
\pdfoutput=1
\usepackage[numbers,sort&compress]{natbib}
\usepackage{blindtext}
\usepackage{graphicx}
\usepackage{dcolumn}
\usepackage{bm}
\usepackage{amssymb}
\usepackage{color}
\usepackage[T1]{fontenc} 
\usepackage{mathrsfs,amsfonts,dsfont}
\usepackage{nccmath}
\usepackage{amstext}
\usepackage{mathtools}
\usepackage{tablefootnote}
\usepackage[inline]{enumitem}
\graphicspath{{graphics/}}
\usepackage[english]{babel}

\usepackage[table]{xcolor}

\usepackage{wrapfig}
\usepackage{lipsum} 
\usepackage{blindtext}

\usepackage{braket}
\usepackage{bbm} 
\usepackage{soul}

\usepackage[nolist,withpage]{acronym}
\makeatletter
\AtBeginDocument{%
  \renewcommand*{\AC@hyperlink}[2]{%
    \begingroup
      \hypersetup{hidelinks}%
      \hyperlink{#1}{#2}%
    \endgroup
  }%
}
\makeatother

\usepackage[colorlinks=true,citecolor=blue,linkcolor=blue]
{hyperref}

\usepackage{amsthm}

\newtheorem{result}{Result}
\newtheorem{result_reset}{Result}
\newtheorem{observation}{Observation}

\newenvironment{proofsketch}
  {\par\pushQED{\qed}\normalfont\topsep6pt \trivlist
   \item[\hskip\labelsep\itshape Proof sketch.]\ignorespaces}
  {\popQED\endtrivlist\addvspace{6pt}}

\newcommand{\addMTQ}[1]{\textcolor{magenta}{#1}}

\DeclareMathOperator{\Tr}{Tr}

\hypersetup{pdftitle = {Strict hierarchy between n-wise measurement simulability, compatibility structures, and multi-copy compatibility},
       pdfauthor = {Lucas Tendick, Costantino Budroni, Marco Túlio Quintino},
       pdfsubject = {Quantum information theory}, 
       pdfkeywords = {quantum, assemblage, simulability, joint measurability, hierarchy, , compatible, theory, incompatibility, 
       ,genuine, n-wise, generalization, quantum measurements, nonlocality, Bell, steering,
       certification, genuinely incompatible, multi-copy, multi copy incompatible
              }
      }

\begin{document}
\title{Strict hierarchy between $n$-wise measurement simulability, compatibility structures, and multi-copy compatibility}

\author[1]{Lucas Tendick}
\email{lucas-amadeus.tendick@inria.fr}
\affiliation[1]{Inria, CPHT, LIX, CNRS, École polytechnique, Institut Polytechnique de Paris, Palaiseau, France}
\orcid{https://orcid.org/0000-0001-5210-6710}
\author[2]{Costantino Budroni}
\affiliation[2]{Department of Physics “E. Fermi”, University of Pisa, Largo B. Pontecorvo 3, 56127 Pisa, Italy}
\orcid{0000-0002-6562-7862}
\author[3]{Marco Túlio Quintino}
\affiliation[3]{Sorbonne Université, CNRS, LIP6, F-75005 Paris, France}
\orcid{0000-0003-1332-3477}

\begin{abstract}

The incompatibility of quantum measurements, i.e. the fact that certain observable quantities cannot be measured jointly is widely regarded as a distinctive quantum feature with important implications for the foundations and the applications of quantum information theory. While the standard incompatibility of multiple measurements has been the focus of attention since the inception of quantum theory, its generalizations, such as measurement simulability, $n$-wise incompatibility, and mulit-copy incompatibility have only been proposed recently. Here, we point out that all these generalizations are differing notions of the question of how many measurements are \emph{genuinely contained} in a measurement device. We then show, that all notions do differ not only in their operational meaning but also mathematically in the set of measurement assemblages they describe. We then fully resolve the relations between these different generalizations, by showing a strict hierarchy between these notions. Hence, we provide a general framework for generalized measurement incompatibility. Finally, we consider the implications our results have for recent works using these different notions. 

\end{abstract}

\maketitle

\tableofcontents

\section{Introduction}

 Measuring a system's properties takes on a central role in quantum theory and denotes one of the most crucial deviations from our classical view of the world~\cite{Wheeler1983, Busch2016}. While a measurement in classical physics simply reveals the pre-existing value of a physical quantity without disturbing the measured system, quantum measurements have to be regarded as an inherently random process~\cite{RevModPhys.89.015004, Acin2016}, which reveals information only at the cost of disturbing the measured system~\cite{Myrvold_2009}. The role of measurements in quantum theory becomes even more puzzling, when not only one but two or more observable quantities are considered simultaneously. First noted by Heisenberg $100$ years ago~\cite{heisenberg1925umdeutung}, the order in which two observable quantities are measured is pivotal and leads generally to different results. Heisenberg's initial observation led shortly after to the formulation of the celebrated Heisenberg-Robertson uncertainty relation~\cite{Heisenberg1927, PhysRev.34.163}, stating that two non-commuting observable quantities cannot be measured simultaneously with arbitrary precision for all quantum systems. This impossibility of measuring certain observable quantities simultaneously, nowadays known as measurement incompatibility, might initially be seen as a limitation or drawback but is in fact one of quantum physics' most striking features.  Measurement incompatibility lies at the heart of many foundational aspects of quantum theory, with an one-to-one relationship to quantum steering~\cite{RevModPhys.92.015001, Cavalcanti2016, PhysRevLett.98.140402,Quintino2014JM,Uola2014JM,PhysRevLett.115.230402},  quantum state discrimination with post-measurement information~\cite{Barnett2009, PhysRevLett.124.120401,PhysRevLett.122.130402,PhysRevLett.122.130404,Oszmaniec2019,PhysRevLett.125.110401,PhysRevLett.125.110402}, semi-quantum prepare-and-measure tasks~\cite{Guerini2019Distributed}, and operational contextuality~\cite{Tavakoli2020Operational}. Also, measurement incompatibility is a required resource for tasks such as Bell nonlocality~\cite{Nonlocality_review, Bell_seminal,Quintino2015Incompatible,PhysRevA.97.012129,Bene2018}, quantum contextuality~\cite{RevModPhys.94.045007},  prepare-and-measure scenarios~\cite{Gallego2010PM,Frenkel2015JM,Egelhaaf2025PJ_JM}, and quantum random access codes~\cite{Carmeli2020, PhysRevLett.125.080403}.
\indent Because of its central importance to both the foundations and applications of quantum theory, measurement incompatibility has attracted widespread attention and has been extensively studied, particularly over the past decade, see~\cite{Heinosaari2016,RevModPhys.95.011003} for recent reviews. One recent development in the field of measurement incompatibility is the effort to generalize its notion~\cite{Guerini2017, Carmeli2016,PhysRevLett.123.180401} in order to gain even deeper insights into the properties of sets of different measurements (from here on called an assemblage). On a high level, measurement incompatibility denotes the effect that multiple measurements on the same state cannot be seen as one effective measurement, from which the outcomes of all measurements can be recovered jointly. In contrast to that, generalized notions of measurement incompatibility study how many of these effective measurements, or alternatively copies of a given quantum state, are minimally necessary to recover the statistics of a given assemblage. This represents a similar development to that in the field of quantum entanglement~\cite{RevModPhys.81.865, Guehne2009, horodecki2024multipartiteentanglement, Luo2020, PhysRevA.103.L060401}, where the study of different notions of multipartite entanglement reveals a rich structure and genuinely new phenomena beyond standard bipartite entanglement. While the several notions of measurement incompatibility have already been applied beyond the initial definition \cite{Tendick2025QuantumCorrelations, PhysRevLett.131.120202, PhysRevA.109.062203, Ducuara2025Multiobject, PhysRevLett.122.130403, das2025operationalapproachclassifyingmeasurement, Li2025Randomness}, their inner relation between each other remains unclear. Even more, at the current state of the field, the different generalizations of measurement incompatibility are used without taking into account other notions or being even aware that competing concepts exists. This results in a very scrambled up picture of the state-of-the-art. \\
\indent In this work, we resolve this issue by presenting one coherent framework for the existing notions of generalized measurement incompatibility and by establishing a strict hierarchy between them. That is, we first point out which different generalizations exist and discuss their operational or geometrical interpretation. Then, we show that these different notions do not only behave operationally differently but describe mathematically different subsets in the set of all quantum measurements and lead generally to different observations. More precisely, we establish a hierarchy of set inclusions between the notions of measurement simulability, compatibility structures, and multi-copy incompatibility. To do so, we establish in Result~\ref{result1-nontechnical}
to Result~\ref{result4-nontechnnical}
the different set relations and, in Result~\ref{result5-nontechnical},
derive a universal lower bound on the critical visibility
for $n$-copy joint measurability on pure input states.
Finally, we discuss the implications of our results for existing works in the literature and discuss future directions for research in the field. \\
\indent This work is structured as follows. In Section~\ref{Sec:MainResults}, we present a summary of our main results in a high-level way to give a quick and accessible overview for the different sub-communities in the field. In Section~\ref{Sec:Preliminaries}, we formally introduce the concept of measurement incompatibility and review the most important notions. In Section~\ref{Sec:Revisiting} we revisit the different generalizations of measurement incompatibility and fix final notations. Then, in Section~\ref{Sec:Results}, we present the formal statements of our results alongside their corresponding proofs. Finally, we discuss the implications of our work in Section~\ref{Sec:Implications} before we conclude with a more general discussion in Section~\ref{Sec:Discussion}.

\section{Summary of main contributions}
\label{Sec:MainResults}

\indent Our main result is a hierarchy of set inclusions, stated in Eq.~\eqref{Eq:MainResult} for the various generalizations of measurement incompatibility that have previously been developed in~\cite{Guerini2017, Carmeli2016,PhysRevLett.123.180401}. On a conceptual level, we point out that these different generalizations are competing models for the ability to conclude that $m$ measurements are \emph{genuinely different} and cannot be reduced (in some operational sense, depending on the model) to fewer measurements. To the best of our knowledge, these different generalizations are currently used without taking into account all alternative models in the existing literature and without a clear operational justification for a specific choice, see e.g.,~\cite{PhysRevA.109.062203, Tendick2025QuantumCorrelations}. Here, we describe precisely the operational differences between the models and formally prove that there exist a hierarchical structure of the sets of measurement assemblages they describe. \\
\indent To give a high-level overview of our main result and its contribution to the field of measurement incompatibility and simulability, suitable for readers from varying backgrounds, we will state our findings here in a non-technical way, divided into smaller results. We illustrate the operational differences between the various generalizations of measurement incompatibility, by summarizing our results in Fig~\ref{Fig:Circuits}. The different competing models that aim to generalize measurement incompatibility and the state-of-the-art will be revisited in Sec.~\ref{Sec:Revisiting}, while the formal proofs corresponding to the statements in Fig.~\ref{Fig:Circuits} and Eq.~\eqref{Eq:MainResult} are given in Section~\ref{Sec:Results}. \\
\indent As we mentioned above, all generalizations of incompatibility have in common that they answer (using a different operational meaning) the question of whether a collection of measurements, a so-called assemblage $ \mathcal{M} = \left\lbrace \left\lbrace M_{a \vert x} \right\rbrace_a \right\rbrace_x  $ with $x \in [m] \coloneqq x \in \left\lbrace 1, \cdots, m \right\rbrace$, can be understood as only $n < m$ measurements. The special case where one is interested in whether $m$ measurements can be understood as a single measurement, i.e., the case $n=1$  describes operationally exactly the notion of standard measurement incompatibility. 
\begin{figure}[t]
\begin{center}
\includegraphics[width=0.95\textwidth]{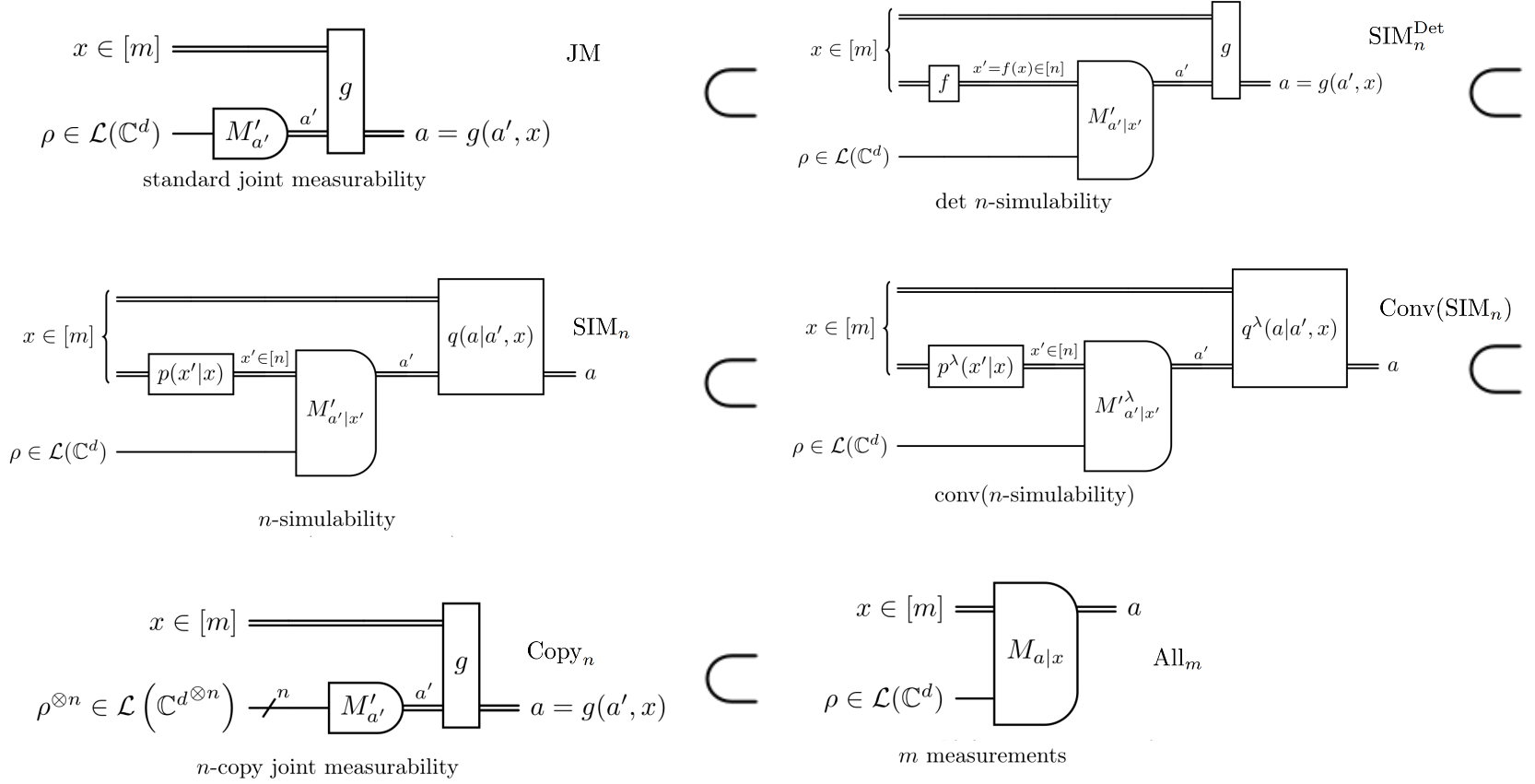}
 \caption{Graphical representation of our main result. We are considering the notion of \emph{standard joint measurability} ($\mathrm{JM}$) and its various generalizations, depicted here in their operational interpretation. In all cases, we consider a set $\mathcal{M}$ of $m$ measurements with inputs $x \in [m]$ and outputs $a$. These measurements are performed on a state $\rho \in \mathcal{L}(\mathds{C}^d)$ and lead to the statistics $p(a \vert x) = \mathrm{Tr}[M_{a \vert x} \rho]$. \emph{Standard joint measurability} asks, whether there exists a single measurement $ \left\lbrace M'_{a'} \right\rbrace_{a'}$, which leads to the same statistics for any state $\rho$, using appropriate post-processing of the measurement outcomes $a'$. As discussed in Eq.~\eqref{Eq:IncompatiblityDefAlternative}, this post-processing can be assumed to be deterministic without loss of generality. \emph{Deterministic $n$-simulability} ($\mathrm{SIM}^{\mathrm{Det}}_n$) asks whether $m$ measurements can be simulated using $n < m$ measurements, and deterministic pre- and post-processing. More generally, \emph{n-simulability} ($\mathrm{SIM}_n$) considers the same question but allows for probabilistic pre- and post-processing. 
 The \emph{conv ($n$-simulability)} model ($\mathrm{Conv}(\mathrm{SIM}_{n})$) allows for (shared) randomness between the pre-processing, the performed quantum measurements, and the post-processing and corresponds geometrically to the convex hull of n-simulability. Finally, \emph{$n$-copy joint measurability} ($\mathrm{Copy}_n$) allows only for a single simulating measurement $ \left\lbrace M'_{a'} \right\rbrace_{a'}$ as in standard joint measurability, however, it is allowed to act on $n$ copies of a state $\rho$ simultaneously. In all circuits, there is a notion of causality from left (earlier) to the right (later). Classical information is represented by double-lines while quantum information, in form of a $d$-dimensional quantum state $\rho \in \mathcal{L}(\mathds{C}^d)$, is represented by a single line. Deterministic functions are denoted by $f,g$, while probabilistic processes are denoted by $p,q$. The \emph{strict} set inclusion $\subset$ between each of the circuits, represents the strict hierarchy of assemblages that can be simulated with a given model, corresponding to Eq.~\eqref{Eq:MainResult}.
 } 
  \label{Fig:Circuits}
\end{center}
\end{figure}
Our main result can be summarized as the following hierarchy of set inclusions, corresponding to the different models depicted in Fig.~\ref{Fig:Circuits}. Namely, we show that it holds
\begin{align}
\label{Eq:MainResult}
 \mathrm{JM} \subset \mathrm{SIM}^{\mathrm{Det}}_{n} \subset \ \mathrm{SIM}_{n} \subset \ \mathrm{Conv}(\mathrm{SIM}_{n}) \equiv \mathrm{JM}^{\mathrm{conv}}_{n} \subset \ \mathrm{Copy}_n \subset \ \mathrm{All}_m.
\end{align}
Here $\mathrm{JM}$ denotes the notion of (standard) joint measureability, while $\mathrm{SIM}^{\mathrm{Det}}_{n}$ and $\mathrm{SIM}_{n}$ denote the sets of measurements that are simulable by $n$ measurements with deterministic or general probabilistic pre-processing, respectively. The set $ \mathrm{Conv}(\mathrm{SIM}_{n}) $ describes the convex hull of the set $\mathrm{SIM}_{n}$, while $ \mathrm{Copy}_n  $ describes those measurement assemblages that are jointly measurable on $n$ copies of any quantum state. Finally, $\mathrm{All}_m$ denotes the set of all assemblages with $ m > n $ measurements. Note that Eq.~\eqref{Eq:MainResult} also includes the set $\mathrm{JM}^{\mathrm{conv}}_{n}$ not depicted in Fig.~\ref{Fig:Circuits} (as it has a clear geometric and less of an operational interpretation), which was first studied in~\cite{PhysRevLett.123.180401} and is shown here to be equal to the set $\mathrm{Conv}(\mathrm{SIM}_{n})$. \\
\indent Our main result is obtained by combining multiple smaller
results (see Result~\ref{result1-formal} to  
Result~\ref{result4-formal}), combined with
relations between some of the sets of interest already
established prior to this work. To clearly lay out our novel findings, we also state these results here in a non-technical way, before stating their technical version, alongside proofs in Section~\ref{Sec:Results}.  Note that for all following results, and the rest of the manuscript in general, we consider the case where $1 < n < m$, unless stated otherwise.

\begin{result} [Insufficiency of deterministic pre-processing]
\label{result1-nontechnical}
The set of measurements that are $n$-simulable, according to the definition in~\cite{Guerini2017}, (see Eq.~\eqref{Eq:SimuablityDef}) i.e., the set-of measurements that can be simulated using $n$ quantum measurements and classical pre-and post-processing
is strictly larger than the set of $n$-simulable measurements with deterministic pre-processing only.
\end{result}
\begin{proofsketch}
To show this result, we present an explicit example of three qubit measurements that can be simulated by two measurements with probabilistic classical pre-processing by stating the simulation model explicitly. Then, we show that no such model exists using deterministic classical pre-processing by optimizing over all possible deterministic strategies (which can explicitly be enumerated, as there is a finite number of deterministic assignments between simulated and simulating measurements). We note that the the proof involves numerical upper bounds for the optimization over all deterministic strategies obtained through the output of an SDP. To convert this into an analytical proof, one would need an explicit construction of a certificate of infeasibility for each SDP. See the formal proof of Result~\ref{result1-formal} for more details. 
\end{proofsketch}

\begin{result}[Non-convexity of the set of $n$-simulable measurements]
\label{result2-nontechnical}
The set of $n$-simulable measurements~\cite{Guerini2017} is not convex, hence being strictly contained in its convex hull.
\end{result}
\begin{proofsketch}
We consider specifically the first non-trivial scenario, i.e., $n=2$ and $m=3$. More specifically, we consider the three noisy Pauli measurements, defined by the observables $\sigma_x, \sigma_y, \sigma_z$ and show that there exists a noise regime in which the noisy measurements cannot be simulated with $2$ measurements, while being in the convex hull of $2$-simulable measurements. The proof that the noisy measurements are in the convex hull is given by explicitly decomposing them as convex combination of $2$-simulable measurements. Then, we show that this same set of noisy Pauli measurements is not $2$-simulable by combining an SDP-based numerical method with an $\epsilon$-net grid approach and an error estimation argument to obtain a rigorous computational result. As with Result~\ref{result1-nontechnical}, it would also require here an additional step to turn this into a fully analytical proof, as one would need an explicit construction of a certificate of infeasibility for each SDP. See the formal proof of Result~\ref{result2-formal} for more details. 
\end{proofsketch}

\begin{result} [Correspondence between measurement simulability and compatibility structures]
\label{result3-nontechnical}
The convex hull of the set of $n$-simulable measurements~\cite{Guerini2017} is precisely the set of $n$-wise compatible measurements (see Eq.~\eqref{Def:Genuine-n-wise-comp-structures}) as introduced in~Ref.~\cite{PhysRevLett.123.180401}. This shows that two definitions that were previously used by the measurement incompatibility and simulability community are strongly linked.
\end{result}
\begin{proofsketch}
We show this result by first understanding the correspondence between deterministic  $n$-simulability (corresponding to the set $ \mathrm{SIM}^{\mathrm{Det}}_{n} $) and the compatibility structures introduced in ~Ref.~\cite{PhysRevLett.123.180401}. Then, we take explicitly the convex hull of the set $\mathrm{SIM}_{n}$ and identify a one-to-one correspondence with the set of $n$-wise compatible measurements, here denoted by $\mathrm{JM}^{\mathrm{conv}}_{n}$.    
\end{proofsketch}

\begin{result} [Connection between $n$-wise compatible measurements and $n$-copy jointly measurable measurements]
\label{result4-nontechnnical}
The set of $n$-wise compatible measurements as introduced in~Ref.~\cite{PhysRevLett.123.180401} is strictly contained in the set of $n$-copy jointly measurable measurements introduced in Ref.~\cite{Carmeli2016} (see Eq.~\eqref{Eq:DefinitionKcopy}).
\end{result}
\begin{proofsketch}
The result follows from using a known connection between $n$-simulability and $n$-copy joint measurability \cite{Carmeli2016,PhysRevA.109.062203}, combined with the connection between $n$-simulability and $n$-wise compatible measurements established via Result~\ref{result3-nontechnical}.
\end{proofsketch}

\begin{result}[A universal lower bound on the $n$ copy joint
measurability for $m$ measurements
on pure input states]
\label{result5-nontechnical}
The critical visibility under depolarizing noise
for the $n$-copy joint measurability
on pure input states
of any set of $m$ measurement settings on a qudit system
of dimension $d$ is lower bounded by
$\eta^*=\frac{n(d+m)}{m(d+n)}$.
\end{result}

\begin{proofsketch}
We obtain this lower bound by using the optimal cloning
machine introduced and studied in~\cite{PhysRevA.58.1827}
and using that noise-free measurements on noisy copies
of a pure state $\rho$ can be seen as
noisy measurements on noise-free copies of $\rho$.
This result directly extends the finding first presented
in Sec.~2.4 of Ref.~\cite{Heinosaari2016} for the case
of standard joint measurability.
\end{proofsketch}

\section{Preliminaries: measurement incompatibility}
\label{Sec:Preliminaries}

 We describe a quantum measurement by a \ac{POVM}, i.e., a set $\left\lbrace M_a \right\rbrace_a$ of operators $0 \leq M_{a} \leq \mathds{1}$ such that $\sum_a M_a = \mathds{1}$. Given a state $\rho$, the probability of obtaining outcome $a$ is given by the Born rule $p(a) = \Tr[M_a \rho]$. A \emph{measurement assemblage} is a set of different \acp{POVM} with operators $ \mathcal{M} \coloneqq \left\lbrace \left\lbrace M_{a \vert x} \right\rbrace_a \right\rbrace_x$, where $x$ denotes the particular measurement $ \mathcal{M}_x \coloneqq \left\lbrace M_{a \vert x} \right\rbrace_a  $. An assemblage $\mathcal{M}$ is called \emph{jointly measurable} (or compatible) if it can be simulated by a single measurement (often called \emph{parent \ac{POVM}}) $\left\lbrace M'_{a'} \right\rbrace_{a'}$ and conditional probabilities $q(a \vert x, a')$ such that 
\begin{align}
\label{Eq:IncompatiblityDef}
M_{a \vert x} = \sum_{a'} q(a \vert x, a') M'_{a'} \ \forall \ a,x,  
\end{align}
and it is called \emph{incompatible} otherwise. Note here, that there is generally no restriction on the number of possible outcomes $a'$ of the parent \ac{POVM} $\left\lbrace M'_{a'} \right\rbrace_{a'}$. \\
\indent Joint measurability can be interpreted as follows (see also Fig.~\ref{Fig:MeasurementIncompatiblity}): instead of actually performing the (seemingly) different measurements of the assemblage $\mathcal{M}$, one can simply perform the single measurement $\left\lbrace M'_{a'} \right\rbrace_{a'}$ and post-process the outcomes $a'$ for different measurement settings $x$, via the distribution $q(a \vert x, a')$. The post-processing $ q(a \vert x, a') $ can w.l.o.g. be chosen to be deterministic, by possibly enlarging the number of outcomes of the parent \ac{POVM} $\left\lbrace M'_{a'} \right\rbrace_{a'}$ and shifting all the randomness from $ q(a \vert x, a') $ into the the parent \ac{POVM} \cite{Ali2009, RevModPhys.95.011003}, in direct analogy to Fine's theorem \cite{PhysRevLett.48.291}. More formally, we can make the identification  (see also Sec. III A in~\cite{Cavalcanti2016})
\begin{align}
\label{Eq:DeterministicIncompatibility}
q(a \vert x,a') = \sum_{a'' = 1}^{N_{\mathrm{Det}}} q(a'' \vert a') \delta_{a,  a''_x},    
\end{align}
where $a''_{x}$ is a function from $[m] \coloneqq \lbrace 1,2, \cdots, m \rbrace $ to $[k] \coloneqq  \lbrace 1, 2, \cdots, k \rbrace $, where $k$ denotes the number of each of the $m$ measurements. More specifically, this allows the identification of $ a'' $ with an outcome string such that $ a'' = [a_{x=1}, a_{x=2}, \cdots, a_{x=m}] $ and $a''(x) = a_x$. It follows that there are  $N_{\mathrm{Det}} = k^m$ deterministic input-output assignments (strategies) for an assemblage of $m$ measurements with $k$ outcomes each. Finally, we can re-define the parent \ac{POVM} as $ M^{''}_{a''} := \sum_{a'}  q(a'' \vert a') M'_{a'}$ such that it holds
\begin{align}
M_{a \vert x} = \sum_{a'' = 1}^{N_{\mathrm{Det}}} \delta_{a,  a''_x}  M^{''}_{a''}. 
\label{Eq:DeterministicIncompatibilityFinal}
\end{align}
Conversely, measurement incompatibility is sometimes also introduced directly through a marginalization condition of the form
\begin{align}
\label{Eq:IncompatiblityDefAlternative}
\sum_{\vec{a} :  a_x = a} M'_{\vec{a}} = M_{a \vert x}, \ \forall a,x,
\end{align}
where $\left\lbrace M'_{\vec{a}} \right\rbrace_{\vec{a}}$ is a parent \ac{POVM} with outcomes $ [a_{x=1}, a_{x=2}, \cdots, a_{x=m}] $ with one \emph{outcome-channel} corresponding to each setting. Comparing Eq.~\eqref{Eq:IncompatiblityDefAlternative} and Eq.~\eqref{Eq:DeterministicIncompatibilityFinal}, it is clear that this formulation is nothing else than the deterministic post-processing identification  above and hence equally characterizes the problem of joint measurability.
\begin{figure}[t]
\begin{center}
\includegraphics[scale=0.7]{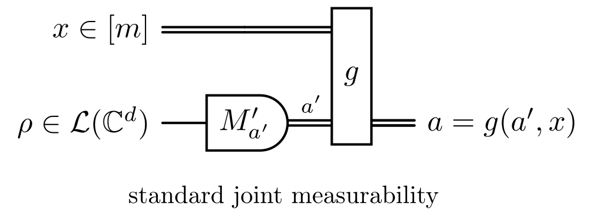}
 \caption{Operational illustration of measurement incompatibility. If $m$ measurement boxes are available, only one of them can be performed at a time. Nevertheless, for jointly measurable assemblages $\mathcal{M}$, there exists a single measurement in form of a parent \ac{POVM} $ \left\lbrace M'_{a'} \right\rbrace_{a'} $, which gives simultaneous access to the outcomes $\left\lbrace a_x \right\rbrace_x$ for all $ x \in [m] \coloneqq \left\lbrace 1,2, \cdots, m \right\rbrace $ upon post-processing the outcome $a'$ with a deterministic function $g(a',x) = \delta_{a, a'_x}$.}
  \label{Fig:MeasurementIncompatiblity}
\end{center}
\end{figure}
This shows that we have two equivalent definitions for standard JM. A definition with probabilistic post-processing, and another where the post-processing is necessarily deterministic. Since they are equivalent, we  will chose which equivalent definition to use up to convenience. \\
\indent While we will work primarily with the definition of joint measurability in form of Eq.~\eqref{Eq:IncompatiblityDef}, the definition in form of Eq.~\eqref{Eq:IncompatiblityDefAlternative} is useful for illustrating examples. Note further that it follows directly from Eq.~\eqref{Eq:IncompatiblityDefAlternative} that the question whether a given assemblage $ \mathcal{M} $ is incompatible, can be solved via a \ac{SDP}, a special instance of convex optimization that is efficiently solvable \cite{boyd_vandenberghe_2004}. Joint measurability can be seen as a generalization of the concept of commutativity from projective measurements to general \acp{POVM}. In fact, if $\mathcal{M}$ describes projective measurements, (in)compatibility reduces to (non-)commutativity. As an illustrative example (presented for instance in \cite{RevModPhys.95.011003}), let us consider noisy Pauli measurements given by
\begin{align}
M^{\eta}_{\pm \vert 1} = \dfrac{1}{2}(\mathds{1} \pm \eta \sigma_x), \ M^{\eta}_{\pm \vert 2} = \dfrac{1}{2}(\mathds{1} \pm \eta \sigma_z),
\end{align}
where $\sigma_x,\sigma_z$ are the Pauli matrices in the $x,z$ direction and $\eta \in [0,1]$ describes the visibility of the measurements. Clearly, for $\eta = 0$ the measurements are jointly measurable, while they are incompatible for $\eta = 1$ by the noncommutativity of projective measurements. To get the exact threshold, let us consider a candidate parent \ac{POVM}, given by
\begin{align}
M^{'\eta}_{i,j} = \dfrac{1}{4}(\mathds{1}+\eta(i \sigma_x + j \sigma_z)),
\end{align}
where $i,j = \pm 1$. It can directly be verified that it holds $M^{\eta}_{\pm \vert 1} = \sum_j M^{' \eta}_{i,j}$ and $M^{\eta}_{\pm \vert 2} = \sum_i M^{' \eta}_{i,j}$ and $\sum_{ij,} M^{' \eta}_{i,j} = \mathds{1}$. However, $M^{' \eta}_{i,j} \geq 0$ is only true for $ \eta \in [0, \tfrac{1}{\sqrt{2}}]$. Hence, $\left\lbrace M^{' \eta}_{i,j} \right\rbrace_{i,j} $ is a parent \ac{POVM} for the noisy $\sigma_x,\sigma_z$ for any $ \eta \in [0, \tfrac{1}{\sqrt{2}}]$. It can be shown (using for instance \ac{SDP} techniques) that the above parent \ac{POVM} is optimal, i.e., $ \mathcal{M}^{\eta} $ becomes incompatible for $ \eta > \tfrac{1}{\sqrt{2}} $ \cite{Busch1986JM,Heinosaari2016, RevModPhys.95.011003}.\\ 
\indent Throughout the manuscript, we will use the visibility $\eta$ to quantify the robustness of a given assemblage $\mathcal{M}$ with respect to different subsets of the set of all quantum measurement assemblages, corresponding to the different models depicted in Fig.~\ref{Fig:Circuits}. That is, we are interested when the noisy assemblage $\mathcal{M}^{\eta}$ defined via
\begin{align}
\label{Eq;NoiseModel}
M^{\eta}_{a \vert x} = \eta M_{a \vert x} + (1-\eta) \Tr[M_{a \vert x}] \dfrac{\mathds{1}}{d},   
\end{align}
is contained in the sets appearing in Eq.~\eqref{Eq:MainResult}, which will be formally introduced in Sec.~\ref{Sec:Revisiting}. Joint measurability and measurement incompatibility as defined via Eq. \eqref{Eq:IncompatiblityDef} are naturally well-defined for more than $2$ measurements. However, by only looking at the incompatibility of $n > 2$ measurements, many of the more detailed properties of $\mathcal{M}$ are not captured and the physical insights are somewhat limited. To exemplify this, let us consider the $m=3$ noisy Pauli measurements $\mathcal{M} = \left\lbrace \lbrace  M^{\eta}_{\pm \vert s} \rbrace_{\pm} \right\rbrace_{s} = \left\lbrace \lbrace \dfrac{1}{2}(\mathds{1} \pm \eta \sigma_s) \rbrace_{\pm} \right\rbrace_s$ for $s \in \left\lbrace x,y,z \right\rbrace$. \\
\indent When we investigate the incompatibility properties of $\mathcal{M}$, we could simply check Eq.~\eqref{Eq:IncompatiblityDef} or Eq.~\eqref{Eq:IncompatiblityDefAlternative}. However, one could also ask when only a subset of two measurements in $\mathcal{M}$ are jointly measurable. Even more, and at the core of this work, one could be interested in making a statement about for which visibility $\eta$ constitutes the assemblage $\mathcal{M}$ \emph{genuinely} $m=3$ measurements and when can it be reduced to effectively two measurements (note that the one measurement case corresponds exactly to the definition of measurement incompatibility). To answer this question, different approaches have been made in the past to generalize measurement incompatibility, such that more of the inherent structure of an assemblage $\mathcal{M}$ is captured and a classification into \emph{genuinely} $m$ measurements can be made. 

\section{Revisiting the different generalizations of measurement incompatibility}
\label{Sec:Revisiting}

\indent If one whishes to consider a more general notion of joint measurability or measurement incompatibility, there are several different models that have been proposed and studied in the literature so far, i.e., there does not exist one unique generalization. The approaches to generalize measurement incompatibility that exist so far in the literature can be categorized roughly into three categories: Measurement simulability \cite{Guerini2017}, compatibility structures \cite{PhysRevLett.123.180401}, and multi-copy compatibility \cite{Carmeli2016}. All of these have in common that they do not only capture in more details the properties of a measurement assemblage $ \mathcal{M} $, but can also be used to infer how many measurements are minimally necessary to reproduce  $ \mathcal{M} $ in some operationally meaningful way. However, as we will show, all these generalizations differ not only in the operational interpretation, but also in the mathematical and consequentially physical characterization. We introduce the three main generalizations of measurement incompatibility (and special instances thereof) in the following.  

\subsection{Measurement simulability}
\label{SubSection:Simulability}
The notion of measurement simulability was first introduced and thoroughly studied in \cite{Guerini2017} (see also \cite{Haapasalo2011QuantumMeasurements, Osmaniec2017Simulating} for prior notions of simulability in more specific scenarios). It is based around the idea that the ability to perform some measurements $ \mathcal{M}' = \left\lbrace \left\lbrace M'_{a'|x'} \right\rbrace_{a'} \right\rbrace_{x'=1, \cdots, n}  $ allows one to simulate a wider class of measurements by classical pre-processing and classical post-processing. More formally, according to \cite{Guerini2017}, we say that the assemblage $ \mathcal{M}' $ can simulate the assemblage $ \mathcal{M} = \left\lbrace \left\lbrace M_{a|x} \right\rbrace_a \right\rbrace_{x = 1, \cdots,m}  $ if there exist probabilities $p(x'|x)$ and $q(a|x,a')$ such that
\begin{align}
\label{Eq:SimuablityDef}
M_{a \vert x} = \sum_{x'=1}^n p(x' \vert x) \sum_{a'}  M'_{a'|x'} \ q(a|x,a'), \ \forall a,x.
\end{align}
Note that this notion of simulability has a clear operational interpretation (see also Fig.~\ref{Fig:MeasurementSimulability}) and is entirely based on classical pre-and post-processing. That is, upon receiving the information that measurement $x$ of $\mathcal{M}$ should be simulated, we perform measurement $x'$ of the assemblage $\mathcal{M}'$ with probability $p(x' \vert x)$. Upon receiving outcome $a'$, we then output outcome $a$ according to the distribution $q(a|x,a')$. It is important to note here, that compared to the notation chosen in \cite{Guerini2017}, we do not condition the post-processing $q(a|x,a')$ specifically also on $x'$, as the information about $x'$ can be encoded into $a'$, since there is no limitation in the number of outcomes of the assemblage. \\
\indent The model is motivated from a perspective of having only access to very limited experimental resource. i.e., being only able to perform the measurements $\mathcal{M}'$ (for instance, one is restricted to certain measurement directions due to limitations in a lab). From the perspective of only being able to perform $\mathcal{M}'$, the definition of simulability in Eq.~\eqref{Eq:SimuablityDef} is already fairly general. However, in some works \cite{PhysRevLett.122.130403, Ducuara2025Multiobject}, a slightly more general definition of simulability was chosen, in which classical randomness between the pre and post-post processing is also allowed. More formally, it means one has access to a random variable $ \Lambda $ which takes on the values $\lambda$ with probability $\pi(\lambda)$ such that the \emph{randomness-assisted} simulability model becomes:
\begin{align}
\label{Eq:SimuablityDefRandomness}
M_{a \vert x} = \sum_{\lambda} \pi(\lambda) \sum_{x'=1}^n p(x' \vert x, \lambda) \sum_{a'}  M'_{a'|x'} \ q(a|x,a', \lambda), \ \forall a,x.
\end{align}
\begin{figure}[t]
\begin{center}
\includegraphics[scale=0.5]{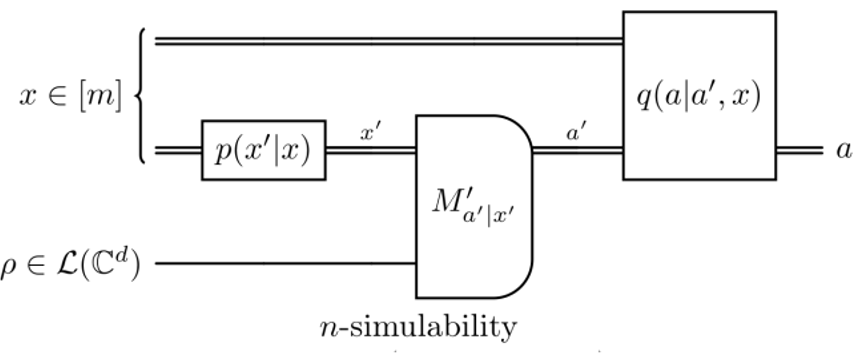}
 \caption{Operational interpretation of measurement $n$-simulability. Given that measurement setting $x$ of $\mathcal{M}$ should be simulated, a measurement $x'$ from the simulating assemblage $ \mathcal{M}' $ is chosen probabilistically. Upon obtaining the outcome $a'$ from the measurement $ \mathcal{M}'_{x'} $, a final outcome $a$ is produced probabilistically, taking the knowledge of $x,a'$ into account. Note that, as there is no limitation on $a'$, the outcome can also encode the choice of measurement $x'$, hence listing it specifically as an argument in the post-processing $q$ would be redundant.}
  \label{Fig:MeasurementSimulability}
\end{center}
\end{figure}
Notably, there is still no dependence of the operators $ M'_{a'|x'} $ on $\lambda$, as we consider $\mathcal{M}' = \left\lbrace \left\lbrace M'_{a'|x'} \right\rbrace_{a'} \right\rbrace_{x'}$ as given. If one also includes the dependence on $\lambda$ of the simulating measurements i.e., considers a simulation strategy given by 
\begin{align}
\label{Eq:SimuablityDefConvHull}
M_{a \vert x} = \sum_{\lambda} \pi(\lambda) \sum_{x'=1}^n p(x' \vert x, \lambda) \sum_{a'}  M^{'\lambda}_{a'|x'} \ q(a|x,a', \lambda), \ \forall a,x,
\end{align}
one obtains formally the convex hull of the simulation model described in Eq.~\eqref{Eq:SimuablityDef}. Note further that the question of whether $ \mathcal{M} $ can be simulated by $ \mathcal{M'} $
can a priori not be decided by a linear program, as the probabilities $p(x' \vert x)$ (respectively $ p(x' \vert x, \lambda) $) and $ q(a|x,a') $ (respectively $ q(a|x,a', \lambda) $) are variables, leading to non-linear constraints.  \\ 
\indent The model of measurement simulability can also be used in situations where the simulating assemblage $ \mathcal{M}' $ is not given, but only some of its properties are fixed. For us most importantly, we can consider a model where $ \mathcal{M}' $ is not given, but it is fixed to contain at most $n$ measurements. In that case, we consider whether an assemblage $ \mathcal{M} $ is \emph{n-simulable}. More formally, one can ask whether $ \mathcal{M} = \left\lbrace \left\lbrace M_{a \vert x} \right\rbrace_a \right\rbrace_x $ can be simulated by $n$ measurements by checking that there exist probabilities $p(x' \vert x)$, $q(a|x,a')$, and \acp{POVM} $\left\lbrace \left\lbrace M'_{a' \vert x'} \right\rbrace_{a'} \right\rbrace_{x' = 1, \cdots, n}$ such that Eq.~\eqref{Eq:SimuablityDef} holds true (or Eq.~\eqref{Eq:SimuablityDefRandomness} in the randomness-assisted verison). In the particular case where $ \mathcal{M}$ contains $m = 3$ measurements, and $x' \in \left\lbrace 1,2 \right\rbrace$ one needs to check that
\begin{align}
\label{Eq:Simuablity2Def}
M_{a \vert x} = \sum_{x'=1}^2 p(x'\vert x) \sum_{a'}  M'_{a'|x'} q(a|x,a'),
\end{align}
holds true for some $p(x' \vert x)$, $q(a|x,a')$, and $ M'_{a'|x'} $. We denote by $ \mathrm{SIM}_{n} $ the set of all assemblages $\mathcal{M}$ that are $n$-simulable. Furthermore, we define $ \mathrm{SIM}^{\lambda}_{n} $ to be the set of all \emph{randomness assisted} $n$-simulable assemblages and we denote by $\mathrm{Conv}(\mathrm{SIM}_{n})$ the convex hull of the set  $ \mathrm{SIM}_{n} $, as defined in Eq.~\eqref{Eq:SimuablityDefConvHull}. \\
\indent Note that measurement simulability is clearly a generalization of measurement incompatibility, as we recover it by considering only a single simulating measurement, i.e., $ n = 1 $. One particular restriction to measurement simulability, which will be of particular use to us, is the case where the pre-processing is limited to deterministic strategies, i.e., deterministic functions $x' = f(x)$ with $p(x' \vert x) = \delta_{x', f(x)} $.  In this case, we will refer to "deterministic measurement simulability". \\
\indent While physically well-motivated, measurement simulability comes with the disadvantage of being hard to analyze, as it is a priori non-linear and even more crucially its definition may not define a convex set, as indeed confirmed by  Result~\ref{result2-nontechnical}. To gain some additional insights on measurement simulability, note that we can rewrite the definition of measurement simulability according to Eq.~\eqref{Eq:SimuablityDef} (for simplicity here with only $n=2$ simulating measurements, but it straightforwardly generalizes) in the following convenient way
\begin{align}
M_{a \vert x} &= \sum_{x' = 1,2} p(x' \vert x) \sum_{a'} q(a \vert x,a') M'_{a' \vert x'} \nonumber \\
&= \sum_{x' = 1,2} p(x' \vert x) F^{\mathrm{JM}, (x')}_{a \vert x},\label{JMAssemblageFormSimulation}
\end{align}
where $ F^{\mathrm{JM}, (x')}_{a \vert x} = \sum_{a'} q(a \vert x,a') M'_{a' \vert x'} $ describes a jointly measurable assemblage $\mathcal{F}^{\mathrm{JM}, (x')}$. \\
\indent While the above expression looks similar to a convex combination of jointly measurable assemblages on first sight, it is quite different from that. In order to simulate $\mathcal{M}$, we are allowed to chose from two measurement assemblages $\mathcal{F}^{\mathrm{JM}, (1)}$ and $\mathcal{F}^{\mathrm{JM}, (2)}$ which are each jointly measurable. However, note that there will be measurements in $\mathcal{F}^{\mathrm{JM}, (1)}$ that are not jointly measurable with measurements in $\mathcal{F}^{\mathrm{JM}, (2)}$. Furthermore, we are allowed to change the mixing of measurements from those assemblages depending on the setting $x$ that we want to simulate (via $p(x' \vert x)$), contrary to a convex mixture where this term would be constant for all $x'$. Note that the above characterization in Eq.~\eqref{JMAssemblageFormSimulation} has already been made in the work~\cite{Guerini2017} (Eq.~$(15)$ therein). \\
\indent Let us finish this sub-section by giving an illustrative example of measurement simulability. To that end, we consider projective measurements corresponding to the $\sigma_x,\sigma_z, H$ operator where $H = \tfrac{1}{\sqrt{2}}\begin{pmatrix}
1 & \phantom{-}1 \\
1 & -1
\end{pmatrix}$ is the Hadamard gate operator. As the $H$ operator lies ``in between'' $  \sigma_x $ and $ \sigma_z $, it might be interesting to know how well all $3$ measurements can be simulated by only performing two measurements. A heuristic numerical optimization shows that the assemblage allows for a
 visibility of (at least) $ \eta_{\mathrm{prob}} \approx 0.8165$. More specifically, we sampled $10^6$ probability distributions $p(x'\vert x)$ and performed the remaining \ac{SDP} optimization to get evidence of the optimal strategy. This resulted in finding no larger visibility than $ \eta_{\mathrm{prob}} \approx 0.8165$. To achieve this visibility, we chose the pre-processing $p(1 \vert 1) = 0,  p(1 \vert 2) = 1, p(1 \vert 3) = \tfrac{1}{2}$, with the other probabilities given by normalization. This seems to be optimal, which is very intuitive, as we use $\sigma_x$ and $\sigma_z$ measurements to simulate noisy versions of themselves and use a convex combination of both directions to simulate the noisy version of $ H $ measurement. 

\subsection{Compatibility structures}

The notion of measurement incompatibility has been generalized in a different direction by considering so-called compatibility structures \cite{PhysRevLett.123.180401}. For instance, more detailed structures of incompatibility, such as pairwise compatibility and genuine triplewise incompatibility, beyond the standard definition of incompatibility have been introduced and studied (see Fig.~\ref{Fig:CompatiblityStructuresGeometrical}). We introduce these compatibility structures first for an assemblage $\mathcal{M}$ in the case of $m=3$ measurements and consider its generalizations to a general $m$ afterward. \\
\indent Let us define the sets $\mathrm{JM}^{(s,t)}$ with $s,t \in \left\lbrace 1,2,3 \right\rbrace$ such that $s \neq t$ as the sets containing assemblages of $m=3$ measurements in which the measurement $s$ and $t$ are jointly measurable. Hence, any assemblage $\mathcal{M} \in \mathrm{JM}^{(s,t)}$ can be seen as effectively $2$ measurements, as there exists a parent \ac{POVM} simultaneously performing the measurements $(s,t)$. The intersection of the sets $\mathrm{JM}^{\mathrm{pair}} \coloneqq \mathrm{JM}^{(1,2)} \cap \mathrm{JM}^{(1,3)} \cap \mathrm{JM}^{(2,3)}$ contains all assemblages in which any pair of two measurements are compatible, the so-called \emph{pairwise} compatible assemblages. On the other hand, the set $\mathrm{JM}_2^{\mathrm{conv}} \coloneqq \mathrm{Conv}(\mathrm{JM}^{(1,2)},\mathrm{JM}^{(1,3)}, \mathrm{JM}^{(2,3)})$ describes the convex hull of the sets $\mathrm{JM}^{(1,2)}$, $\mathrm{JM}^{(1,3)}$, and $\mathrm{JM}^{(2,3)}$, i.e., it contains all assemblage that can be written as a convex combination of assemblages where one pair of measurements is compatible. More formally, it contains all assemblages of the form
\begin{align}
\mathcal{M} = p_{(1,2)} \mathcal{J}^{(1,2)} + p_{(1,3)} \mathcal{J}^{(1,3)}+p_{(2,3)} \mathcal{J}^{(2,3)}, \label{ConvexCombOfIncompStructures}
\end{align}
where $\mathcal{J}^{(s,t)} \in \mathrm{JM}^{(s,t)}$ and the convex combination is to be understood on the level of the individual \ac{POVM} effects. The set $ \mathrm{JM}_2^{\mathrm{conv}} $, which we will also refer to as set of $2$-wise compatible assemblages, can be understood as follows. It contains all assemblages $\mathcal{M}$ that can probabilistically be seen as effectively $2$ measurements, across a certain \emph{compatibility cut}, similarly to the notion of bi-separable states in entanglement theory. Geometrically, the set $ \mathrm{JM}_2^{\mathrm{conv}} $ can be understood as the largest set depicted in Fig.~\ref{Fig:CompatiblityStructuresGeometrical}. 
\begin{figure}[t]
\begin{center}
\includegraphics[scale=0.35]{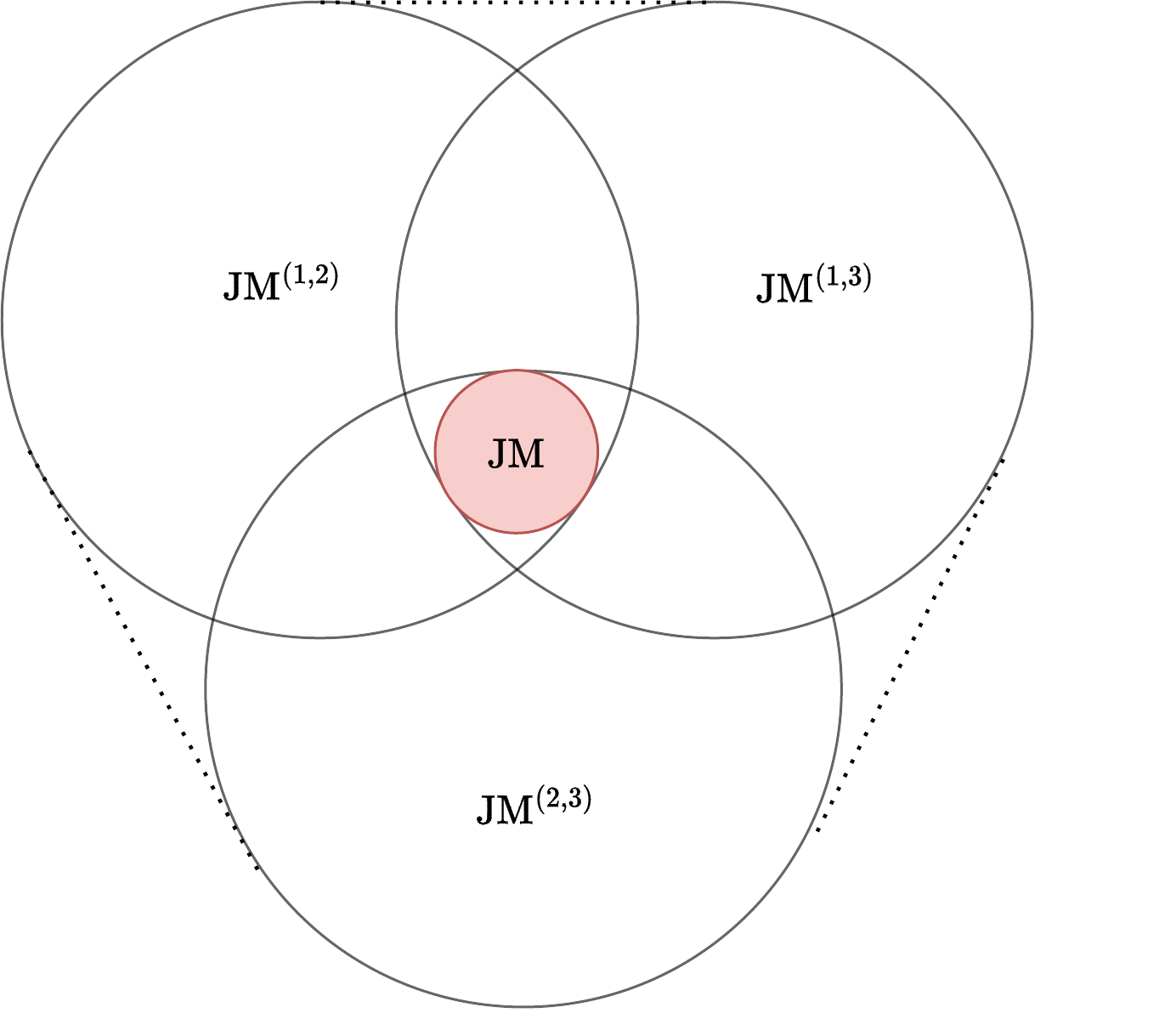}
 \caption{Geometrical representation of the compatibility structures introduces in Ref.~\cite{PhysRevLett.123.180401}. The sets $ \mathrm{JM}^{(s,t)}$ contain assemblages in which the pair $(s,t)$ is compatible. In the intersection of the sets $\mathrm{JM}^{(1,2)}$, $\mathrm{JM}^{(2,3)}$, $\mathrm{JM}^{(1,3)}$ lies the set $\mathrm{JM}$ as standard joint measurability as a proper subset. The convex hull $\mathrm{JM}_2^{\mathrm{conv}} $ of the sets $\mathrm{JM}^{(1,2)},\mathrm{JM}^{(1,3)}, \mathrm{JM}^{(2,3)}$ corresponds to the set of $2$-wise compatible measurements.}
  \label{Fig:CompatiblityStructuresGeometrical}
\end{center}
\end{figure} \\
Finally an assemblage $\mathcal{M} \notin \mathrm{JM}_2^{\mathrm{conv}}$ is said to be \emph{genuinely triplewise incompatible}. More precisely, the idea is that any assemblage that is not contained in the sets $\mathrm{JM}^{(1,2)}, \mathrm{JM}^{(2,3)}, \mathrm{JM}^{(1,3)}$ and even convex combinations of assemblages from these sets must necessarily be an assemblage of genuinely three measurements, i.e., it cannot be seen as probabilistically $2$ measurements. \\
\indent To get more familiar with the concept of compatibility structures, we take a look at the example given in \cite{PhysRevLett.123.180401}. We consider again the three projective measurements $\mathcal{M}^{\eta}= \left\lbrace \left\lbrace M^{\eta}_{a \vert x} \right\rbrace_a \right\rbrace_x$ which represent the Pauli $\sigma_x,\sigma_y,\sigma_z$ observables subjected to depolarizing noise, i.e., we analyze the incompatibility of the assemblage $\mathcal{M}^{\eta}$ defined via
\begin{align}
M^{\eta}_{\pm \vert s} = \dfrac{1}{2}(\mathds{1} \pm \eta \sigma_s), \label{DepolarizedPaulis}
\end{align}
where $s \in \left\lbrace x,y,z \right\rbrace$ and $\eta$ is the visibility. It was shown in \cite{PhysRevLett.123.180401} that $\mathcal{M}^{\eta} \in \mathrm{JM}_2^{\mathrm{conv}}$ iff $\eta  \leq  \tfrac{\sqrt{2}+1}{3} \approx 0.80$. Thus, for $ \eta \in \left[\tfrac{1}{\sqrt{2}}, \tfrac{\sqrt{2}+1}{3}\right] $, the assemblage  $\mathcal{M}^{\eta}$ is not pairwise compatible for any pair, i.e., 
$\mathcal{M}^{\eta} \notin \mathrm{JM}^{(1,2)}, \mathrm{JM}^{(2,3)}, \mathrm{JM}^{(1,3)}$ while it is also not genuinely triplewise incompatible. \\
\indent Interestingly, the simulation strategy used in \cite{PhysRevLett.123.180401} to reach the critical visibility of $\eta = \tfrac{\sqrt{2}+1}{3}$ relies on a convex combination of strategies that include noise-free Pauli measurements in the $\sigma_x,\sigma_y,\sigma_z$ directions. This raises the question about the role of the free randomness in terms of convex combinations. On the one hand it is very natural to assume these to be free. However, if the simulation of noisy $\sigma_x,\sigma_y,\sigma_z$ measurements involves probabilistically performing noise-free measurements in the same direction this can be problematic. More specifically, one could argue from a strictly preparation-inspired point that an experimentalist who wants to simulate the statistics of the assemblage $\mathcal{M}^{\eta}$ with $ \eta = \tfrac{\sqrt{2}+1}{3} $ according to the strategy in \cite{PhysRevLett.123.180401}, needs access to the noise-free $\sigma_x,\sigma_y,\sigma_z$ measurements to simulate noisy versions of them, making the simulation itself obsolete.\\ 
\indent While the authors in \cite{PhysRevLett.123.180401} considered general compatibility structures, represented by arbitrary hypergraphs, we focus here the notion of $n$-wise compatibility  as a generalization of $2$-wise compatibility (i.e., not being genuine triplewise incompatible). Interestingly, this requires a slight generalization of the definitions made in \cite{PhysRevLett.123.180401}, whenever $m \neq n+1$. For simplicity, let us first consider the case where we have the assemblage $ \mathcal{M} $ of $m=n+1$ measurements and we are interested in its $n$-wise compatibility. Then, using a notation similar to that from \cite{Tendick2025QuantumCorrelations}, we say the assemblage $\mathcal{M}$ is $n$-wise compatible if it does admit a decomposition 
\begin{align}
\label{DefnGenuineIncomp-m-wise}
M_{a \vert x} =   \sum_{(t_1,t_2 \neq t_1) \in \mathcal{T}}   p_{(t_1,t_2)}  J_{a \vert x}^{(t_1,t_2)},
\end{align}
where $\left\lbrace \left\lbrace J_{a \vert x}^{(t_1,t_2)} \right\rbrace_a \right\rbrace_{x = 1, \cdots, m} \in \mathrm{JM}^{(t_1,t_2)}  $ are assemblages of $m$ measurements in which the measurements $(t_1,t_2)$ are jointly measurable and $p_{(t_1,t_2)}$ denotes the probability with which the assemblage $\mathcal{J}^{(t_1,t_2)}$ is employed. Here, the set $\mathcal{T}$ contains all tuples $ (t_1,t_2 \neq t_1) $ constructed from the integers $[m] = \left\lbrace 1, \cdots, m \right\rbrace$.

\begin{figure}[h]
\begin{center}
\includegraphics[scale=0.5]{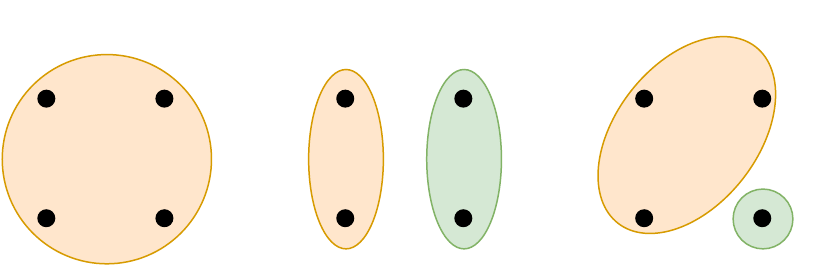}
 \caption{Example of the possible compatibility hypergraphs to simulate $m=4$ measurements with $n=2$ measurements, according to the approach of compatibility structures. Each dot represents one measurement. We omitted labelings of the measurements here, as these hypergraphs only represent all possibilities up to permutations of measurements. All measurements that are contained in an hyper-edge are jointly measurable. The scenario on the very left corresponds to standard joint measurability, where one of the parent \acp{POVM} is simply ignored. In the middle, two parent \acp{POVM} are used to simulate two measurements each. The case on the right represents the scenario where $3$ measurements are jointly measurable, while the remaining measurement is simulated by itself. To define $n$-wise compatibility, all possibilities of such partitions of the measurements from an assemblage $\mathcal{M}$ into (at most) $n$ disjoint subsets (represented by the hypergraphs) have to be considered. We then allow for convex combinations between different partitions, which represents the probabilistic mixture of different strategies to simulate $\mathcal{M}$ with (at most) $n$ measurements. Note that the cases with less than $n$ disjoint subsets are already included in those with exactly $n$ subsets, as the measurements from disjoint hyper-edges might still be compatible with each other.   
 }
  \label{Fig:CompatibilityHyperGraphsExample}
\end{center}
\end{figure}

\indent While it is straightforward to define $n$-wise compatibility in the case $n = m-1$ through Eq.~\eqref{DefnGenuineIncomp-m-wise}, it becomes more subtle for a general $n < m$. To exemplify this, let us consider the situation when $m=4$ and $n=2$ (see also Fig.~\ref{Fig:CompatibilityHyperGraphsExample}). That is, we have $2$ parent \acp{POVM} available to (deterministically) simulate $4$ measurements. This can happen in two non-trivial ways: First, $3$ out of the $4$ measurements might be simulated by a single parent \ac{POVM}, while the remaining measurement is simulated by the second parent \ac{POVM} (which can be the measurement itself). Second, the $4$ measurements might be split up into two disjunct pairs of two measurements, each of which is simulated separately by a parent \ac{POVM}. Interestingly, these two ways to simulate (deterministically) $4$ measurements with $2$ measurements are, a priori, not comparable, i.e., it is not obvious that one strategy would be strictly stronger than the other. Note that there also exist a third (but trivial) way to simulate $4$ measurements by $2$, which is to fully ignore one of the available parent \acp{POVM}, reducing the strategy to standard joint measurability. \\
\indent To formalize general $n$-wise compatibility in the following, we invoke the language of compatibility hypergraphs \cite{PhysRevLett.123.180401}, which we will use in a slightly more general notion here. In a compatibility hypergraph, each measurement $x$ is represented by a vertex and all measurements contained in the same hyper-edge are jointly measurable with each other. To then define $n$-wise compatibility, we consider all partitions $C_i$ of $m$ measurements into $n$ disjoint subsets. For example, in the above case (see Figure~\ref{Fig:CompatibilityHyperGraphsExample}), $C = [(1,2,3),(4)]$ represents the situation in which the measurements $x=1,2,3$ are jointly measurable and measurement $x=4$ is trivially jointly measurable with itself. On the other hand, $\tilde{C} = [(1,2),(3,4)]$ represents the situation where the measurements $x=1,2$ and simultaneously the measurements $ x = 3,4 $ are jointly measurable. In total, all allowed partitions $C_i$ for the case of $2$-wise compatibility of $m = 4$ measurements are then given by  
\begin{align}
\mathcal{C}_{n\to m} = \left\lbrace C_i \right\rbrace_i = \big\lbrace [(1,2,3),(4)], [(1,2,4),(3)], [(1,3,4),(2)], [(2,3,4),(1)], \nonumber \\
[(1,2),(3,4)], [(1,3),(2,4)], [(1,4),(2,3)] \big \rbrace,     \label{Eq:PartitionsExample}
\end{align}
which trivially includes the case $ [(1,2,3,4)] $ of full joint measurability. Then, we say $\mathcal{M}$ is $2$-wise compatible if it can be written as a convex combination over all allowed partitions of $\mathcal{M}$ into at most $2$ disjoint subsets, i.e.,
\begin{align}
M_{a \vert x} = \sum_i p_i J^{C_i}_{a \vert x}.   \label{Def:Genuine-n-wise-comp-structures}
\end{align}
Now, in the most general case, we say that an assemblage $\mathcal{M}$ of $m$ measurements is $n$-wise jointly measurable 
if $\mathcal{M}$ obeys a decomposition like in Eq.~\eqref{Def:Genuine-n-wise-comp-structures} where each of the partitions $C_i$ of $m$ measurements describes a collection of $n$ disjoint subsets (represented by the hyper-edges) which are jointly measurable. Note that any assemblage of $m$ measurements is trivially $m$-wise jointly measurable. \\

\subsection{Multi-copy incompatibility}

The final model, as a generalization of measurement incompatibility, that we consider is inspired by the idea that two measurements $ \mathcal{M}_1 = \left\lbrace M_{a \vert 1} \right\rbrace_a $ and $ \mathcal{M}_2 = \left\lbrace M_{a \vert 2} \right\rbrace_a $ are only incompatible because we demand the statistics to originate from just a single copy of a state $\rho$. If one instead allows for multiple copies, say $n=2$ copies of $\rho$, any measurements $\mathcal{M}_1,\mathcal{M}_2$ are trivially compatible, as they can simply be performed on different copies of the same state. However, this does not exclude the possibility of assemblages with $m \geq 3$ measurements to still be incompatible, even on $2$ copies of $\rho$. This approach in generalizing incompatibility does strictly speaking not directly quantify how many measurements are genuinely needed to implement some assemblage $\mathcal{M}$, it counts instead the number of quantum states that is necessary. More precisely, following \cite{Carmeli2016}, an assemblage $\mathcal{M}$ is said to be $n$-copy jointly measuarble if it holds 
\begin{align}
\label{Eq:DefinitionKcopy}
\mathrm{Tr}[M_{a \vert x} \rho] = \sum_{a'} p(a \vert x, a') \mathrm{Tr}[M'_{a'} \rho^{\otimes n}] \ \ \forall \rho,
\end{align}
where $\left\lbrace M'_{a'} \right\rbrace_{a'}$ is now a parent \ac{POVM} acting on $n$-copies of a state $\rho$ jointly.  The concept of $n$-copy compatibility is operationally depicted in Fig.~\ref{Fig:KCopyIncompatibility}. 
\begin{figure}[t]
\begin{center}
\includegraphics[scale=0.6]{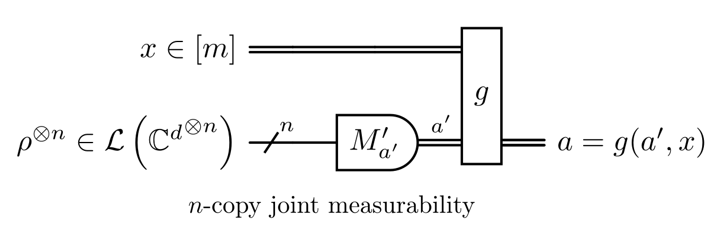}
 \caption{Operational interpretation of multi-copy joint measurability. The measurement box gets $n$ copies as a quantum input  and is allowed to perform a distributed (possibly entangled) measurement on all $n$ copies of $\rho$ to reproduce the statistics of a given assemblage $\mathcal{M}$.}
  \label{Fig:KCopyIncompatibility}
\end{center}
\end{figure} We remark here, that following the identification in Eq.~\eqref{Eq:DeterministicIncompatibility} the probabilistic post-processing $ p(a \vert x, a') $ can w.l.o.g. be replaced by a deterministic assignment $ a = g(x,a') $ by enlarging the outcome space of $a'$, i.e., both definitions (with deterministic or probabilistic pre-processing) lead to equivalent observations. For the $3$ noisy Pauli measurements, it can be shown that the critical visibility which still allows them to be measured jointly on two copies of any state $\rho$ is given by $\eta = \sqrt{3}/2$ . Note that any assemblage $\mathcal{M}$ of $m$ measurements is $m$-copy jointly measurable, simply by performing one measurement each per copy. 
We denote by $ \mathrm{Copy}_n $ the set of all assemblages that are jointly measurable on $n$ copies. \\
\indent Note that while it is not directly clear from Eq.~\eqref{Eq:DefinitionKcopy}, it follows from~\cite{Carmeli2016} (see also~\cite{PhysRevA.109.062203}) that the problem of deciding whether a given assemblage $\mathcal{M}$ is jointly measurable on $n$ copies, can be decided via an \ac{SDP}. Note further that $n$-copy joint measurability is intuitively directly related to the notion of cloning, or rather the inability to clone quantum systems perfectly. If one would have the ability to clone a quantum state $\rho$ perfectly, one simply could perform the measurements on the output of this cloning machine on each copy of $\rho$.

\section{Results}
\label{Sec:Results}

In this section, we will give a formal statement of our results and prove Result~\ref{result1-nontechnical} to Result~\ref{result5-nontechnical}, leading to the hierarchy of set inclusions given in Eq.~\eqref{Eq:MainResult}. Our results allow for a simultaneous comparison of all the measurement incompatibility generalizations revisited in Sec.~\ref{Sec:Revisiting}. For the sake of clarity, we focus in the presentation on the first non-trivial scenario, i.e, assemblages $\mathcal{M}$ of $m=3$ settings and $n =2 <m $ but do provide a proof for the general case where necessary.

\subsection{Deterministic vs. probabilistic simulability}

Our first result concerns the notion of measurement simulability as introduced in Eq.~\eqref{Eq:SimuablityDef}. In particular, we focus here on a simulability process in which the pre-processing, i.e., the probabilities $ p(x' \vert x) $ to simulate measurement $x$ using setting $x'$ are deterministic. That is, the probabilities $ p(x' \vert x) \in \left\lbrace 0, 1 \right\rbrace $ reduce to  $x' = f(x)$, where $f:[m]\to[n]$ is a function. Note again  that the  post-processing can be chosen to be deterministic w.l.o.g. (see Section~\ref{Sec:Preliminaries} and Section~\ref{Sec:Revisiting}). Let us denote by $ \mathrm{SIM}^{\mathrm{Det}}_{n} $ the set of all assemblages simulable by $n$ measurements with deterministic pre-processing. Before proving that $ \mathrm{SIM}^{\mathrm{Det}}_{n} $ is a strict subset of all $n$-simulable assemblages, i.e., $ \mathrm{SIM}^{\mathrm{Det}}_{n} \subset  \mathrm{SIM}_{n}$ we first make the following useful observation. 

\begin{observation} [Simulability with deterministic pre-processing]
\label{observation1}
Consider the measurement simulability process as described in Eq.~\eqref{Eq:SimuablityDef} for an assemblage $\mathcal{M}$ with $3$ measurement settings, i.e., $x \in \left\lbrace 1,2,3 \right\rbrace$. There exists a $2$-simulability strategy with deterministic pre-processing for $\mathcal{M}$, i.e., $\mathcal{M} \in   \mathrm{SIM}^{\mathrm{Det}}_{2} $ if and only if $\mathcal{M}$ is contained in the union of the sets $ \mathrm{JM}^{(1,2)}, \mathrm{JM}^{(1,3)}, \mathrm{JM}^{(2,3)}$. Remember that the sets $\mathrm{JM}^{(s,t)}$ are defined under Eq.~\eqref{ConvexCombOfIncompStructures} as those assemblages (of $m=3$ measurements) in which the pair $s,t$ is compatible. 
\end{observation}
\begin{proof}
A deterministic pre-processing means that the simulators do not mix for a particular setting and measurement $x' = 1$ of the assemblage $\mathcal{M}'$ which simulates $\mathcal{M}$ according to
\begin{align}
M_{a \vert x} = \sum_{x'=1}^2 \delta_{x',f(x)} \sum_{a'}  M'_{a'|x'} \ q(a|x,a'), \ \forall a,x,    
\end{align}
is used to simulate either one, two, three, or none of the measurements of $\mathcal{M}$ deterministically (the same for $x'=2$ accordingly). In the case $x' = 1$ is never selected, we always use $x'=2$, which reduces to standard joint measurability. Similarly, when $x' = 1$ is always selected, we never use the setting $x'=2$, also reducing the consideration to standard measurement incompatibility. When one for the simulating settings, say $x' = 1$, is selected to simulate one of the measurements, while the other, say $x'=2$, is used to simulate two measurements, we recover exactly the definition of the sets $\mathrm{JM}^{(s,t)}$. Since any of the combinations (either of the settings $(1,2)$, $(1,3)$, $(2,3)$) can be chosen to be simulated by one parent \ac{POVM} (while the remaining measurement can be simulated by itself as the second parent \ac{POVM}), we obtain exactly the union of $ \mathrm{JM}^{(1,2)}, \mathrm{JM}^{(1,3)}, \mathrm{JM}^{(2,3)}$. 
\end{proof}

Having Observation~\ref{observation1} at hand, we can directly proceed and proof Result~\ref{result1-formal}, which shows that probabilistic pre-processing is generally useful.

\begin{result_reset} [Insufficiency of deterministic pre-processing]
\label{result1-formal}
 There exists some assemblages $\mathcal{M}$ with $m=3$ measurements that are $2$-simulable, i.e., $\mathcal{M} \in \mathrm{SIM}_{2}$ which cannot be simulated with deterministic pre-processing, i.e., $\mathcal{M} \notin \mathrm{SIM}^{\mathrm{Det}}_{2}$. Hence, it holds generally that the set $\mathrm{SIM}_{n}$ of measurements that are $n$-simulable, is strictly larger than the set $\mathrm{SIM}^{\mathrm{Det}}_{n}$ of $n$-simulable measurements with deterministic pre-processing.
\end{result_reset}

\begin{proof}
Note that to make the statement for a general $n$, it is enough to find some $n$ where the sets $\mathrm{SIM}_{n}$ and $\mathrm{SIM}^{\mathrm{Det}}_{n}$ do not coincide. We consider a specific assemblage $ \mathcal{M} $ of $m=3$ measurements. In particular, we are considering projective measurements according to the eigenstates of $\sigma_x,\sigma_z$ and the Hadamard operator $H = \tfrac{1}{\sqrt{2}}\begin{pmatrix}
1 & \phantom{-}1 \\
1 & -1
\end{pmatrix}$. This example was already introduced in Section~\ref{SubSection:Simulability}. We consider the depolarizing noise robustness of $\mathcal{M}$ according to Eq.~\eqref{Eq;NoiseModel} and compare deterministic with probabilistic pre-processing. The deterministic critical visibility can be found easily by running through all deterministic strategies (which can easily be enumerated, as there is a finite number of strategies to simulate $3$ settings with $2$ measurements via deterministic pre-processing), which reduces the remaining optimization (over the simulating assemblage $\mathcal{M}'$) to an \ac{SDP}. For the probabilistic pre-processing we chose a particular strategy, given by $p(1 \vert 1) = 0, \ p(1 \vert 2) = 1, \ p(1 \vert 3) = \tfrac{1}{2}$, with the other probabilities given by normalization. This choice seems optimal after a heuristic numerical exploration, as already discussed in Section~\ref{SubSection:Simulability}. More specifically, we sampled $ 10^6 $ probability distributions $p(x'\vert x)$ and performed the remaining \ac{SDP} optimization to get an intuition of whether our analytical strategy is optimal. Note that for the gap between deterministic and probabilistic pre-processing to be proven, optimality of the probabilistic pre-processing strategy is not necessary. We find critical visibilities of $\eta_{\mathrm{det}} \approx 0.7654$ and $ \eta_{\mathrm{prob}} \approx 0.8165$. Hence, $\mathrm{SIM}^{\mathrm{Det}}_{n} \subset \mathrm{SIM}_{n}$. \\
\indent Note that the proof is based on the implicit assumption that the error in the evaluation of the bound for deterministic simulability via SDP, obtained with solver tolerances (relative duality gap, primal and dual
feasibility) of $10^{-8}$, is much lower than the observed gap, which is approximately 0.05. To complete this argument to an analytical proof, one would need an explicit construction of a certificate of infeasibility for each SDP of the deterministic simulation (one for each deterministic strategy) for $\eta > 0.7654$. See also Obs.~\ref{Observation2}.

\end{proof}

\subsection{Simulability vs. simulability with correlated classical pre-and post-processing.}

Before going on with proving our main results, we want to discuss a nuanced extension to measurement simulability as defined in Eq.~\eqref{Eq:SimuablityDef} to the \emph{randomness assisted} model in Eq.~\eqref{Eq:SimuablityDefRandomness}. In the latter, one has access to a random variable $\Lambda$ taking on values $\lambda$ which can be used to coordinate (correlate) the pre-processing and the post-processing. While it is clear that the model including randomness is at least as general as the model without, it is not obvious that it becomes more powerful. Mathematically speaking, the question becomes whether $ \mathrm{SIM}_{n} \subset   \mathrm{SIM}^{\lambda}_{n} $. While we do not have a decisive answer to this question, we want to argue here that it is likely the case. In the case it holds $ \mathrm{SIM}_{n} = \mathrm{SIM}^{\lambda}_{n} $, what formally needs to be shown is that, for any assemblage $\mathcal{M}$ for which a decomposition 
\begin{align}
M_{a \vert x} = \sum_{\lambda} \pi(\lambda) \sum_{x'=1}^n p(x' \vert x, \lambda) \sum_{a'}  M'_{a'|x'} \ q(a|x,a', \lambda), \ \forall a,x,
\end{align}
exists, there also exists a decomposition of the form 
\begin{align}
\label{Eq:SimulabilityRandomnessArgument}
M_{a \vert x} = \sum_{x'=1}^n \tilde{p}(x' \vert x) \sum_{a'}  \tilde{M}'_{a'|x'} \ \tilde{q}(a|x,a'), \ \forall a,x.
\end{align}
\indent While it is unproblematic to introduce some additional variable $\lambda$ in the post-processing of Eq.~\eqref{Eq:SimulabilityRandomnessArgument} by artificially increasing the number of outcomes $a'$, it is not possible to also introduce a similar dependence in the pre-processing. The reason for this is that the number of settings $x$ of the assemblage $\mathcal{M}$ are fixed by design (as $\mathcal{M}$ is given) and also the number of measurement settings $x'$ of $\tilde{\mathcal{M}}'$ is bounded from above as we allow maximally $n$ measurements. Hence, there is generally no direct way to re-write the model into a \emph{randomness assisted} model.  What is currently preventing a conclusive answer is the problem that both sets, i.e., $ \mathrm{SIM}_{n}$ and $ \mathrm{SIM}^{\lambda}_{n} $ do not admit a known characterization that allows to find a global optimum through any efficient optimization method.

\subsection{Non-convexity of the set of $n$-simulable assemblages}
\label{Sec:Nonconvexity}

While the concept of measurement simulability has been studied to quite some extent in \cite{Guerini2017}, there is surprisingly little known about its \emph{geometric structure}. In particular, while it is suspected that the question of deciding whether an assemblage is $n$-simulable is not directly decidable by an SDP, it is unclear whether this is indeed the case. More crucially, it is not even known whether the set  $ \mathrm{SIM}_{n}$ is convex or not for $n > 1$. Answering this questions is not only important from an optimization perspective but also contributes to further understanding the operational meaning of measurement simulability. In the following, we show that the sets $ \mathrm{SIM}_{n}$ are, in general, not convex, by explicitly proving the non-convexity of $ \mathrm{SIM}_{2}$.

\begin{result_reset}[Non-convexity of the set of $n$-simulable measurements]
\label{result2-formal}
The assemblage $\mathcal{M}^{\eta}$ containing $m=3$ measurements, corresponding to the noisy Pauli $\sigma_x, \sigma_y, \sigma_z$ observables, according to the noise model in Eq.~\eqref{Eq;NoiseModel}, satisfies $\mathcal{M}^{\eta} \in \mathrm{Conv}(\mathrm{SIM}_{2})$ for $ \eta \leq \tfrac{(\sqrt{2} + 1)}{3} \approx 0.804$.  On the other hand\footnote{Also, numerical evidences suggest that $ \mathcal{M}^{\eta} \notin \mathrm{SIM}_{2}$, if and only if $\eta>\tfrac{1}{\sqrt{2}}\approx 0.7071$.},  $ \mathcal{M}^{\eta} \notin \mathrm{SIM}_{2}$, for $ \eta = \tfrac{(\sqrt{2} + 1)}{3} \approx 0.804$. This implies that $\mathrm{SIM}_{2}$ is strictly contained in $\mathrm{Conv}(\mathrm{SIM}_{2})$.
\end{result_reset}

\begin{proof}
An assemblage $\mathcal{M}$ is said to be $2$-simulable if it admits a decomposition of the form
\begin{align}\label{eq:2-sim}
M_{a \vert x} = \sum_{x'=1}^2 p(x'\vert x) \sum_{a'}  M'_{a'|x'} q(a|x,a'),
\end{align}
for some assemblage $\mathcal{M}'$ with $n=2$ measurements and pre-and post-processing probabilities $  p(x'\vert x)$, $q(a|x,a')$. 
Without loss of generality, we can take the post-processing $q(a|x,a')$ to be deterministic, as all randomness can be absorbed in the definition of $\mathcal{M}'$. So $q(a|x,a')$ is considered to be fixed in a canonical way, i.e., $q(a|x,a')=\delta_{a, a'_x}$ and $a' = [a_1',a_2',a_3']$, see, e.g., \cite{PhysRevLett.123.180401}. 

Deciding $2$-simulability of  $\mathcal{M}^{\eta}$ can, thus, be defined in terms of the following optimization problems
\begin{equation}\label{eq:n-sim_SDP}
\begin{split}
\text{\bf Given }  \mathcal{M}^\eta, p \ \ \ &\\
\text{\bf find }  \nu_p = \min_{\mathcal{M}', \lambda}\ & \sum_{a,x} \lambda_{a|x},\\
\text{ subject to} & -\lambda_{a|x} \openone \leq M_{a|x}^\eta - \sum_{x'a'} p(x'|x)  q(a|x,a') M'_{a'|x'} \leq \lambda_{a|x} \openone,\  \forall \ a,x\\
& M_{a'|x'} \geq 0, \ \forall \ a',x',\ \ \sum_{a'} M'_{a'|x'} = \openone, \ \forall x',\\
\end{split}
\end{equation}
and, for fixed $\mathcal{M}^\eta$,
\begin{equation}\label{eq:n-sim_opt}
\nu^* = \min_{p \in \mathcal{P}} \nu_p,
\end{equation}
where $\mathcal{P}=\{ \vec{p} \in \mathbb{R}^{nm}\ |\ p(x'|x)\geq 0, \ \sum_{x'} p(x'|x)=1\}$ denotes the space of distributions for the pre-processings. It is straightforward to see that $\nu^* =0 \iff  \mathcal{M}^{\eta} \in \mathrm{SIM}_{2}$. While Eq.~\eqref{eq:n-sim_SDP} is an SDP, Eq.~\eqref{eq:n-sim_opt} cannot be directly formulated as an SDP, since it involves products of $p$ and $\mathcal{M}'$. The idea is to estimate the optimal value of problem \eqref{eq:n-sim_opt}, by approximating  $\mathcal{P}$ with a discrete grid and estimating the error obtained in doing so. First, notice that  $\nu^*$ coincides with a minimal (norm) distance between $\mathcal{M}^\eta$ and $\mathrm{SIM}_{2}$, i.e.,  $\nu^* =\min_{\mathcal{M}_s \in \mathrm{SIM}_{2}} \| \mathcal{M}^\eta - \mathcal{M}_s\|_{A}$, where we define the norm $\|\cdot \|_{A}$ of an assemblage (see also \cite{Tendick2023}) as $\|\mathcal{M} \|_{A}:= \sum_{a,x} \|M_{a|x} \|$, where $\| \cdot\|$ is the operator norm.

Let us consider  the optimal solution $(p^*, \mathcal{M}^{\prime *})$ of Eq.~\eqref{eq:n-sim_opt}. It defines the minimal distance $\nu^*$ between the original assemblage  $\mathcal{M}^\eta$ and the simulated one $\mathcal{M}_s^*$, constructed from $(p^*, \mathcal{M}^{\prime *})$ according to Eq.~\eqref{eq:2-sim}. Now, consider a finite set $\mathcal{P}_g \subset \mathcal{P}$ and let us approximate the distribution $p^*$ with $p_0 \in \mathcal{P}_g$ and construct the corresponding simulated assemblage  $\mathcal{M}_s^{g,0}$ by inserting $p_0$ and $ \mathcal{M}^{\prime *}$ in  Eq.~\eqref{eq:2-sim}. We observe that by choosing for $\mathcal{P}_g$ a hypercubic lattice of step-size $\ell$ , we can find $p_0 \in \mathcal{P}_g$ such that $|p^*(x'|x)-p_0(x'|x)| \leq \ell/2$ for all $x,x'$, i.e., every component of the probability vector is at most at distance $\ell/2$ from the same component of a grid point. We have that $(p_0, \mathcal{M}^{\prime *})$ is a feasible solution of the optimization problem on the grid, i.e., the problem $\min_{p \in \mathcal{P}_g} \nu_p$. We define the value of the corresponding  objective function as $\nu_{g,0}$ and estimate the error we obtain when approximating $p^*$ with $p_0$, namely,
\begin{equation}
\nu_{g,0} := \| \mathcal{M}^\eta - \mathcal{M}_s^{g,0}\|_{A} \leq  \| \mathcal{M}^\eta - \mathcal{M}_s^*\|_{A} + \| \mathcal{M}_s^{g,0} - \mathcal{M}_s^*\|_{A} \leq \nu^* +\varepsilon,
\end{equation}
where $\varepsilon$ represents any upper bound on the quantity $\| \mathcal{M}_s^{g,0} - \mathcal{M}_s^*\|_{A}$. 
The optimal solution of the problem on the grid, i.e., $\nu^*_g := \min_{p \in \mathcal{P}_g} \nu_p$, can in principle improve the value of  $\nu_{g,0}$ by using an assemblage different from $\mathcal{M}^{\prime *}$. We thus have
\begin{equation}
\nu^*_g \leq \nu_{g,0} \leq \nu^* +\varepsilon\ \ \Longrightarrow \ \ \nu^*_g -\varepsilon \leq  \nu^*.
\end{equation} 
To conclude our argument, it is then sufficient to find a grid $\mathcal{P}_g$ and a good estimate of the error $\varepsilon$ such that $\nu^*_g -\varepsilon >0$, which, in turns implies $\nu^* > 0$ and consequently that $\mathcal{M}^{\eta} \notin \mathrm{SIM}_{2}$. Note that $\nu^*_g$ can be explicitly computed by computing an SDP for each grid point $p \in \mathcal{P}_g$. 
Let us now estimate $\| \mathcal{M}_s^{g,0} - \mathcal{M}_s^*\|_{A}$. We have
\begin{equation}\label{eq:grid_approx}
\begin{split}
\| \mathcal{M}_s^{g,0} - \mathcal{M}_s^*\|_{A} = \sum_{a,x} \| \sum_{a'x'} p^*(x'|x)  q(a|x,a') M_{a'|x'}^{\prime *} - \sum_{a'x'} p_0(x'|x)  q(a|x,a') M_{a'|x'}^{\prime *}\| \\
 = \sum_{a,x} \| \sum_{a'x'}  \left[ p^*(x'|x) - p_0(x'|x) \right]  \ q(a|x,a') M_{a'|x'}^{\prime *}\|
\leq \frac{\ell}{2} \sum_{a,x} \| \sum_{a'x'}   q(a|x,a')    M_{a'|x'}^{\prime *} \| \\
\leq \frac{\ell}{2} \sum_{a,x} \| \sum_{a'x'}    M_{a'|x'}^{\prime *} \| = \frac{\ell}{2} |\mathcal{A}| |\mathcal{X}| |\mathcal{X}'| 
\end{split}
\end{equation}
where we used the estimate of $|p^*(x'|x)-p_0(x'|x)| $ previously computed, the fact that $M_{a'|x'}^{\prime *}$ is positive semidefinite\footnote{For $X_i \geq 0$ if $-\delta \leq \alpha_i \leq \delta $ for all $i$, we have $-\delta X_i\leq \alpha_i X_i \leq \delta X_i$, thus $- \sum_i \delta X_i \leq \sum_i \alpha_i X_i \leq \sum_i \delta X_i$, and, finally, $\| \sum_i \alpha_i X_i \| \leq \|\sum_i \delta X_i\|$.} and that $\sum_{a'}    M_{a'|x'}^{\prime *} =\openone$, and we defined  $|\mathcal{A}|, |\mathcal{X}|$ and  $|\mathcal{X}'|$ as the cardinalities of the sets of possible values for $a,x,$ and $x'$. Inserting all the parameters in Eq.~\eqref{eq:grid_approx}, we obtain
\begin{equation}
\varepsilon := \frac{\ell}{2} |\mathcal{A}| |\mathcal{X}| |\mathcal{X}'|  = \frac{\ell}{2}\ 2 \times 3 \times 2 = 6 \ell.
\end{equation}
To compute explicitly $\nu^*_g$, we solve the SDP  in Eq.~\eqref{eq:n-sim_SDP} on a grid of $50$ points for each direction, giving $\ell = 1/50$. Note that each $p$ is defined by three parameters between $0$ and $1$, since $p(0|x)=1-p(1|x)$ for all $x$. It is, thus, sufficient to evaluate the SDP in Eq.~\eqref{eq:n-sim_SDP} on a three-dimensional cubic lattice of $50^3$ grid points. For $ \eta = \tfrac{(\sqrt{2} + 1)}{3} $, the numerical calculations give
\begin{equation}
\nu^*_g  \approx 0.1953, \ \varepsilon = \frac{6}{50} = 0.12 \ \ \Longrightarrow \ \ \nu^* \geq \nu^*_g  - \varepsilon \approx 0.1953 - 0.12 = 0.0753>0.
\end{equation}
Note that, even if the evaluation of $\nu^*_g $ contains a numerical error, this is several orders of magnitude smaller than the gap $0.0753$ (see also Observation~\ref{Observation2}). We can conclusively claim that $\mathcal{M}^{\eta} \notin \mathrm{SIM}_{2}$ for $ \eta = \tfrac{(\sqrt{2} + 1)}{3}$ and, therefore, that $\mathrm{SIM}_{2}$ is strictly contained in ${\rm Conv}(\mathrm{SIM}_{2})$.
\end{proof}

Note that the above described $\epsilon $-net method can not only be used to show that $\mathrm{SIM}_{2}$ (or more generally $\mathrm{SIM}_{n}$) is non-convex but also to estimate generically how well a particular assemblage $\mathcal{M}$ can be simulated by a fixed number of measurements. 

\begin{observation}
\label{Observation2}

We remark that the validity of the proof of Result \ref{result2-formal} relies on the implicit assumption that the numerical errors originating from the use of floating-point arithmetic are smaller than the gap $\nu^*_g  - \varepsilon$. Considering that the solver tolerances (relative duality gap, primal and dual feasibility) are of the order of $10^{-8}$ and the obtained gap $\nu^*_g  - \varepsilon$ is of the order of $10^{-1}$, this seems a reasonable assumption. Indeed, this assumption is commonly used when dealing with convex optimization methods in quantum information theory. 
To obtain a completely rigorous proof, one should derive upper and lower bounds for $\nu^*_g$  by providing feasible solutions for the primal and dual problem that are verified by exact rational arithmetic. This is indeed possible following the construction from Ref.~\cite{Bavaresco2021Strict}. However, we argue that doing this additional step to obtain a completely rigorous computer assisted proof would not lead us to any additional physical insight of the problem itself. \\
\indent We also remark at this point, that a similar analysis can be applied to Result~\ref{result1-formal}. That is, one could in principle, convert the evident gap between the numerically obtained bound for the simulability with deterministic pre-processing and the analytically obtained bound for probabilistic pre-processing into a fully analytical proof.

\end{observation}

\subsection{The convex hull of the $n$-simulability set and its connection to compatibility structures}

Since the set $ \mathrm{SIM}_{n} $ is not convex, it is interesting to ask what the convex hull of the set physically and geometrically corresponds to. Interestingly, we show that the set $ \mathrm{Conv}(\mathrm{SIM}_{n}) $ corresponds to a different generalization of measurement incompatibility, previously studied in \cite{PhysRevLett.123.180401} under the name \emph{compatibility structures}. That is, the sets  $ \mathrm{Conv}(\mathrm{SIM}_{n}) $ and $ \mathrm{JM}^{\mathrm{conv}}_{n} $ coincide.

\begin{result_reset} [Correspondence between measurement simulability and compatibility structures]
\label{result3-formal}

The convex hull of the set $\mathrm{SIM}_{n}$, i.e., the set of all assemblages $\mathcal{M}$ that admit a decomposition of the form
\begin{align}
\label{Eq:DefConvHull}
M_{a \vert x} = \sum_{\lambda} \pi(\lambda) \sum_{x'=1}^{n} p^{\lambda}(x' \vert x) \sum_{a'} q^{\lambda}(a \vert a',x) M'^{\lambda}_{a' \vert x'},     
\end{align}
denoted by  $ \mathrm{Conv}(\mathrm{SIM}_{n}) $, coincides with the set of $n$-wise compatible assemblages $ \mathrm{JM}^{\mathrm{conv}}_{n} $ defined according to Eq.~\eqref{Def:Genuine-n-wise-comp-structures}, i.e., $ \mathrm{Conv}(\mathrm{SIM}_{n}) =  \mathrm{JM}^{\mathrm{conv}}_{n} $. 
\end{result_reset}

\begin{proof}

To make the proof as clear as possible, we first consider the special but most relevant case $n=m-1$, i.e., the scenario where one is interested in whether an assemblage $\mathcal{M}$ of $m$ measurements contains \emph{genuinely $m$ measurements} or less than $m$ measurements. In that case, the set $ \mathrm{Conv}(\mathrm{SIM}_{(m-1)}) $ defines assemblages of the form
\begin{align}
M_{a \vert x} = \sum_{\lambda} \pi(\lambda) \sum_{x'=1}^{m-1} p^{\lambda}(x' \vert x) \sum_{a'} q^{\lambda}(a \vert a',x) M'^{\lambda}_{a' \vert x'}.  \label{Eq:ConvSimulability-m-1}   
\end{align}
\indent Similarly to Observation~\ref{observation1}, let us first focus on the case where the pre-processing $p^{\lambda}(x' \vert x)$ is deterministic, i.e., it reduces to a function $x' = f^{\lambda}(x)$. It is direct to see that w.l.o.g. we can assume the functions $f^{\lambda}(x) = x'$ to be surjective, as a simulation strategy which ignores a specific simulating measurement $\tilde{x'}$ can always be replaced by one where it is included. Then, as we consider the case of $(m-1)$-simulability with deterministic pre-processing, it follows that no more than two measurement settings $t_1,t_2 \neq t_1$ are simulated by the same $x'$. Hence, the only strategy that we have to consider is that one pair of measurements $(t_1,t_2)$
is simulated by a single measurement while the remaining $m-2$ measurements of the assemblage $\mathcal{M}$ are trivially simulated by the remaining $m-2$ simulating measurements of the assemblage $\mathcal{M}'$ (which can be the measurements themselves). That means, with deterministic pre-processings, we will always describe assemblages $\mathcal{J}^{(t_1,t_2)}$, where the measurements $(t_1,t_2 \neq t_1)$ are jointly measurable. Now, as $(t_1,t_2 \neq t_1)$ can be chosen freely, and we have access to convex combinations of assemblages $\mathcal{J}^{(t_1,t_2)}$ with different $(t_1,t_2 \neq t_1)$, we recover exactly the definition of $(m-1)$-wise compatible measurements, according to Eq.~\eqref{DefnGenuineIncomp-m-wise}.\\ 
\indent In the case of probabilistic pre-processings, let us decompose the probabilities $ p^{\lambda}(x' \vert x)$ as follows:
\begin{align}
p^{\lambda}(x' \vert x) \coloneqq p(x' \vert x, \lambda) = \sum_{\lambda'} \pi(\lambda' \vert \lambda) \delta_{x', f^{\lambda'}(x)} ,      
\end{align}
where $f^{\lambda'}(x) : X \mapsto X'$ is a deterministic function and $\pi(\lambda' \vert \lambda)$ the probability that $\lambda'$ (corresponding to the random variable $ \Lambda' $) is chosen, given that $\lambda$ was observed as outcome of the random variable $\Lambda$. We then find that 
\begin{align}
M_{a \vert x} &= \sum_{\lambda, \lambda'} \pi(\lambda) \pi(\lambda' \vert \lambda) \sum_{x'=1}^{m-1} \delta_{x', f^{\lambda'}(x)} \sum_{a'} M'^{\lambda}_{a' \vert x'} q^{\lambda}(a \vert a',x),  \\
&= \sum_{\mu} \pi(\mu) \sum_{x' = 1}^{m-1} \delta_{x', f^{\mu}(x)} \sum_{a'} M'^{\mu}_{a' \vert x'} q^{\mu}(a \vert a',x), \label{Eq:Simulability-renaming-mu}
\end{align}
where we defined $ \mu = (\lambda, \lambda') $ and used that $\pi(\lambda) \pi(\lambda' \vert \lambda) = \pi(\lambda, \lambda')$. Hence, we can always restrict to deterministic functions for the pre-processings, due to the randomness available in all terms on the RHS of Eq.~\eqref{Eq:ConvSimulability-m-1}. Note that the re-naming of the dependency to $\mu$ for terms that do only depend on $\lambda$ or $\lambda'$ is justified as $\mu$ unambiguously defines both $\lambda, \lambda'$ and using $\mu$ directly would make the model only more powerful, in the sense that it potentially allow for the simulation of a larger set of assemblages. This is however, not the case as Eq. \eqref{Eq:Simulability-renaming-mu} corresponds exactly to the definition of non-genuinely $m$-wise incompatible measurements, according to Eq.~\eqref{DefnGenuineIncomp-m-wise}. We therefore conclude that $ \mathrm{Conv}(\mathrm{SIM}_{(m-1)}) =  \mathrm{JM}^{\mathrm{conv}}_{(m-1)} $. \\
\indent With this special case as a set-up, the generalization to the case of arbitrary $m$, $n < m$ is straightforward. In analogy to Eq.~\eqref{Eq:Simulability-renaming-mu}, we can write
\begin{align}
M_{a \vert x}   = \sum_{\mu} \pi(\mu) \sum_{x' = 1}^{n} \delta_{x', f^{\mu}(x)} \sum_{a'} M'^{\mu}_{a' \vert x'} q^{\mu}(a \vert a',x),  \label{Eq:Simulability-renaming-mu-general}
\end{align}
as definition for all assemblages $\mathcal{M}$ that are contained in the set $\mathrm{Conv}(\mathrm{SIM}_{n})$, i.e., in the convex hull of all $n$-simulable assemblages. Let us fix one specific $ \mu = \mu^* $, the term 
\begin{align}
J^{C_{\mu*}}_{a \vert x} = \sum_{x' = 1}^{n} \delta_{x', f^{\mu^*}(x)} \sum_{a'} M'^{\mu^*}_{a' \vert x'} q^{\mu^*}(a \vert a',x),    
\end{align}
describes the deterministic simulation of an auxiliary assemblage $ \lbrace \lbrace J^{C_{\mu*}}_{a \vert x} \rbrace_{a} \rbrace_x $ according to the partition $  C_{\mu*} $ of the $m$ measurements from $ \lbrace \lbrace J^{C_{\mu*}}_{a \vert x} \rbrace_{a} \rbrace_x $  according to the strategy $ \mu^* $, as given in the definition Eq.~\eqref{Def:Genuine-n-wise-comp-structures}. That is, the random variable $\mu$ is simply selecting which deterministic assignment from $x$ to $x'$ is chosen. Clearly, each possible partition $C_{\mu}$ corresponds to a specific assignment $\mu = \mu^*$ of simulation measurements $x'$ given a measurement $x$, according to the deterministic strategy $\delta_{x', f^{\mu^*}(x)}$. Note that there is no advantage in considering the same assignment from $x$ to $x'$ for different $\mu = \mu_1$ and $\mu = \mu_2$, analog to convex combinations of the same compatibility structure being unnecessary. Importantly, for any possible partition $ {C_{\mu*}} $ in the definition of $n$-wise joint measurability (see Eq.~\eqref{Def:Genuine-n-wise-comp-structures} and also Eq.~\eqref{Eq:PartitionsExample}) there is a deterministic strategy $\delta_{x', f^{\mu^*}(x)}$ achieving this partition and vice versa. Hence, Eq.~\eqref{Eq:Simulability-renaming-mu-general} can be understood as the convex hull off all possible partitions and hence it holds $ \mathrm{Conv}(\mathrm{SIM}_{n}) =  \mathrm{JM}^{\mathrm{conv}}_{n} $.

\end{proof}

Our finding in Result~\ref{result3-formal} implies that two, a priori different generalizations of measurement incompatibility that have previously been studied \cite{Guerini2017, PhysRevLett.123.180401}, have a close connection through their convex hull. Note here, that the above proof is general enough that it can easily be adapted to formulate any compatibility hypergraph considered in \cite{PhysRevLett.123.180401} in terms of the measurement simulability framework. This allows for a more compact way to formally capture the physical constraints implied by any compatibility hypergraph. 

\subsection{Genuine $n$-wise incompatibility and its connection to $n$-copy incompatibility}

As introduced in Sec~\ref{Sec:Revisiting}, $n$-copy incompatibility is yet another generalization of measurement incompatibility with a clear operational interpretation. While it does not allow directly to count how many measurements are genuinely contained in an assemblage $\mathcal{M}$, it allows to count the number of copies of any state $\rho$ that is required for it to be jointly measurable. Intuitively, having many copies of a state $\rho$ available is very powerful for information processing. It is hence important to ask, where $n$-copy joint measurability ranks in the hierarchy of set inclusions, according to Eq.~\eqref{Eq:MainResult}. As we will establish in our next result, $n$-copy joint measurability is indeed the strongest (currently known) generalization of measurement incompatibility, i.e., the set $\mathrm{Copy}_{n} $ constitutes the largest subset of the set of all measurement assemblages $ \mathrm{All}_{m} $ discussed in this work.

\begin{result_reset} [Connection between $n$-wise compatible measurements and $n$-copy jointly measurable measurements]
\label{result4-formal}

The set of all $n$-wise compatible assemblages $\mathcal{M}$, i.e., the set $\mathrm{JM}^{\mathrm{conv}}_n$ is strictly contained in the set $\mathrm{Copy}_n$ of all $n$-copy jointly measurable assemblages. 

\end{result_reset}

\begin{proof}

\indent Our proof is based on two ingredients. First, it was shown in~\cite{PhysRevA.97.062102} that $n$-simulability, according to Eq.~\eqref{Eq:SimuablityDef} implies $n$-copy joint measureability according to Eq.~\eqref{Eq:DefinitionKcopy}. Note that while the post-processing in~\cite{PhysRevA.97.062102} depends explicitly on the simulating measurement $x'$, this can always be incorporated in our notation by letting the outcome $a'$ also carry the information $x'$, effectively recovering the definition used in~\cite{PhysRevA.97.062102}. Even more, the resulting parent \ac{POVM} acting on the $n$-copies can always be chosen to have separable \ac{POVM} effects across the $n$ copies. Hence, $\mathcal{M} \in \mathrm{SIM}_n \implies \mathcal{M} \in \mathrm{Copy}_n$. Note that while the set $ \mathrm{SIM}_n $ is generally non-convex (see Result~\ref{result2-formal}), the set $ \mathrm{Copy}_n $ is convex by design. Second, we have shown in Result~\ref{result3-formal} that the convex hull of the set $\mathrm{SIM}_n$, i.e., the set $\mathrm{Conv}(\mathrm{SIM}_{n}) $ is equal to the set $\mathrm{JM}^{\mathrm{conv}}_n$ of $n$-wise compatible measurements. Now, as $\mathrm{JM}^{\mathrm{conv}}_n \equiv \mathrm{Conv}(\mathrm{SIM}_{n}) $, is the smallest convex set containing $\mathrm{SIM}_n$ and  $ \mathrm{Copy}_n$ is convex, it follows that $ \mathrm{JM}^{\mathrm{conv}}_n \subseteq \mathrm{Copy}_n$. Finally, we know from Section~\ref{Sec:Revisiting} that for the three noisy Pauli measurements, the critical $2$-wise joint measurability visibility is given by $ \eta_{ \mathrm{JM}^{\mathrm{conv}}_2} = \tfrac{\sqrt{2}+1}{3}$, while the $2$-copy joint measurability visibility for the Pauli measurements is given by $ \eta_{\mathrm{Copy}_2} = \tfrac{\sqrt{3}}{2} > \eta_{ \mathrm{JM}^{\mathrm{conv}}_2}$, hence it holds $ \mathrm{JM}^{\mathrm{conv}}_n \subset \mathrm{Copy}_n$.

\end{proof}

\subsection{A lower bound on the $n$-copy joint measurability of any $m$ measurements on pure states} \label{subsec:ncopy_bound_pure}

Our final result concerns again the notion of $n$-copy joint measurability and in particular its connection to optimal cloning \cite{PhysRevA.58.1827}. Here we consider the case where the input state is promised to be pure. While cloning or broadcasting quantum states is known to be impossible, an approximate cloning/broadcasting can be done. More precisely, the so-called ``Werner's cloning''~\cite{PhysRevA.58.1827} admits a quantum channel $\mathcal{C}_{n\to m}:\mathcal{L}(\mathbb{C}^d)^{\otimes n}\to \mathcal{L}(\mathbb{C}^d)^{\otimes m}$ such that, for every pure quantum state $\rho=\lvert\psi\rangle\langle\psi\rvert\in \mathcal{L}(\mathbb{C}^d)$ and every $x\in\{1,\ldots,m\}$,
\begin{align} \label{eq:WernerCloningPure}
    \Tr_{\overline{x}}\Big(\mathcal{C}_{n\to m}(\rho^{\otimes n})\Big) = c_n(d,m) \rho + (1-c_n(d,m))\frac{\mathds{1}}{d},
\end{align}
where $c_n(d,m):=\frac{n(d+m)}{m(d+n)}$ and $\Tr_{\overline{x}}$ stands for the partial trace in every subsystem except from $x$. Werner's cloning map is originally defined on the symmetric input subspace. Here $\mathcal{C}_{n\to m}$ denotes a CPTP extension to the full input space, obtained by first measuring the projector onto the symmetric subspace: on that subspace we apply Werner's map, while on its orthogonal complement we prepare a fixed output state. This extension does not affect $\rho^{\otimes n}$ when $\rho$ is pure.

In Result~\ref{result5-formal-pure}, we make use of $n\to m$ Werner's cloning to derive a universal lower bound on the critical visibility of $n$-copy joint measurability on pure states of any assemblage $\mathcal{M}$. We remark that the connection between $1\to m$ cloning and joint measurability in the case where one has a single copy of the state, i.e., $n=1$, was noticed in \cite{Heinosaari2016} and was used there to derive a universal bound on the critical visibility of standard joint measurability. We now extend these ideas to obtain a universal bound on the critical visibility of $n$-copy joint measurability on pure states. For this subsection, we restrict the condition in Eq.~\eqref{Eq:DefinitionKcopy} to pure input states: a set of measurements $\{M_{a|x}\}$ is $n$-copy jointly measurable on pure states if it decomposes as
\begin{align} \label{Eq:DefinitionKcopyPure}
\mathrm{Tr}[M_{a \vert x} \rho] = \sum_{a'} p(a \vert x, a') \mathrm{Tr}[M'_{a'} \rho^{\otimes n}] \ \ \forall \rho=\lvert\psi\rangle\langle\psi\rvert,\quad \langle\psi\vert\psi\rangle=1.
\end{align}
The parent POVM $\{M'_{a'}\}$ and the post-processing probabilities $p(a\vert x,a')$ must be the same for every pure input state. We denote the corresponding critical visibility by $\eta_{\mathrm{Copy}^{\mathrm{pure}}_n}$ to distinguish it from the all-state definition in Eq.~\eqref{Eq:DefinitionKcopy}.

\begin{result_reset} [A universal lower bound on the $n$ copy joint measurability for $m$ measurements on pure states]
\label{result5-formal-pure}

Let $\mathcal{M}$ be any assemblage of $m$ $d$-dimensional quantum measurements and let $1\leq n\leq m$. The critical visibility $\eta_{\mathrm{Copy}^{\mathrm{pure}}_n}$ for $\mathcal{M}$, subjected to the depolarizing-noise model in Eq.~\eqref{Eq;NoiseModel}, to become jointly measurable on $n$ copies of any pure state $\rho$ is lower bounded such that
\begin{align}
 \eta_{\mathrm{Copy}^{\mathrm{pure}}_n} \geq \dfrac{n(d+m)}{m(d+n)}.
\end{align}
\end{result_reset}
\begin{proof}
Let $\{M_{a|x}\}_{x=1}^m$ be a set of $m$ POVMs in dimension $d$. Our proof consists in presenting an explicit joint measurement on pure input states for the set of noisy POVMs $\{M^\eta_{a|x}\}_{x=1}^m$, where $M^\eta_{a|x}:=\eta M_{a|x} + (1-\eta) \Tr(M_{a|x}) \frac{\mathds{1}}{d}$, and $\eta = \dfrac{n(d+m)}{m(d+n)}$. Let $\mathcal{C}_{n\to m}:\mathcal{L}(\mathbb{C}^d)^{\otimes n} \to \mathcal{L}(\mathbb{C}^d)^{\otimes m}$ be the CPTP extension of Werner's cloning channel described above (see Eq.~\eqref{eq:WernerCloningPure}) and $\mathcal{C}^\dagger_{n\to m}:\mathcal{L}(\mathbb{C}^d)^{\otimes m} \to \mathcal{L}(\mathbb{C}^d)^{\otimes n}$ its adjoint map.
Let us define the operator $M'_{a_1,\ldots,a_m}\in \mathcal{L}(\mathbb{C}^d)^{\otimes n}$,
\begin{align}
   M'_{a_1,\ldots,a_m}:= \mathcal{C}_{n\to m}^\dagger (M_{a_1|1} \otimes \ldots \otimes M_{a_m|m}).
\end{align}
Notice that the set $\{M'_{a_1,\ldots,a_m}\}$ forms a valid POVM, since $\mathcal{C}_{n\to m}$ is a valid channel, it is CPTP, hence, its adjoint map is CP and unital, ensuring that
\begin{align}
     M'_{a_1,\ldots,a_m}\geq&0 \\
    \sum_{a_1\ldots a_m} M'_{a_1,\ldots,a_m} =& \mathds{1}.
\end{align}
By definition, a set of POVMs is $n$-copy JM on pure states if it fulfills
\begin{align}
\mathrm{Tr}[M_{a \vert x} \rho] = \sum_{a'} p(a \vert x, a') \mathrm{Tr}[M'_{a'} \rho^{\otimes n}] \ \ \forall \rho=\lvert\psi\rangle\langle\psi\rvert,\quad \langle\psi\vert\psi\rangle=1.
\end{align}
We now set $a':=[a_1,\ldots,a_m]$, so that $M'_{a'}:=M'_{a_1,\ldots,a_m}$. We also set $p(a \vert x, a')$ as the deterministic response function $p(a \vert x, a'):=\delta_{a,a_x}$. In this way, for every pure state $\rho$, we have that
\begin{align}
 \sum_{a'} p(a \vert x, a') \mathrm{Tr}[M'_{a'} \rho^{\otimes n}] =& \sum_{a' : a_x = a} \Tr\Big[\mathcal{C}_{n\to m}^\dagger (M_{a_1|1} \otimes \ldots \otimes M_{a_m|m}) \rho^{\otimes n} \Big], \\
   =& \sum_{a' : a_x = a} \Tr\Big[ (M_{a_1|1} \otimes \ldots \otimes M_{a_m|m}) \mathcal{C}_{n\to m}(\rho^{\otimes n}) \Big], \\
   =& \Tr\Big[ M_{a|x} \otimes \mathds{1}^{\overline{x}} \mathcal{C}_{n\to m}(\rho^{\otimes n}) \Big], \\
   =& \Tr\Big[ M_{a|x} \Tr_{\overline{x}}\Big[\mathcal{C}_{n\to m}(\rho^{\otimes n})\Big] \Big], \\
   =& \Tr\Big[ M_{a|x} \Big( \eta \rho + (1-\eta) \frac{\mathds{1}}{d} \Big) \Big], \\
   =& \Tr\Big[ \Big( \eta M_{a|x} + (1-\eta) \Tr[M_{a|x}] \frac{\mathds{1}}{d} \Big) \rho \Big], \\
   =& \Tr\Big[ M^\eta_{a|x} \rho \Big],
\end{align}
where $\overline{x}$ denotes the complement of $x$ and $\mathds{1}^{\overline{x}}$ denotes the identity on all systems but system $x$.

\end{proof}

We note that in the case of the $2$ copy joint measurability on pure states of $3$ qubit measurements, we obtain a bound of $c_2(2,3)=\tfrac{5}{6}\approx0.833$. For comparison, the critical visibility of the three Pauli measurements under the all-state definition is $\eta_{\mathrm{Copy}_2}=\sqrt{3}/2\approx0.866$.

\section{Implications following from our results}
\label{Sec:Implications}

Our results have direct implications for findings that have already been established in the literature prior to this work. In particular, our results allow us to improve and refine previously established results. 

\subsection{Implications for the semi-device-independent certification of number of measurements}

We consider first the work in \cite{PhysRevA.109.062203}. There, the authors asked the question whether the number of measurements used in an experiment can be certified in a semi-device-independent scenario. That is, the authors consider a steering scenario, in which Alice's measurement device is uncharacterized and Bob has full control over his system and measurement device. In particular, Bob can perform state tomography on his share of the quantum state. By sharing a quantum state $\rho$, Alice will prepare the steering assemblage $ \Vec{\sigma} =  \left\lbrace \left\lbrace \sigma_{a \vert x} \right\rbrace_a \right\rbrace_{x} $ given by
\begin{align}
\sigma_{a \vert x} = \mathrm{Tr}_A[(M_{a \vert x} \otimes \mathds{1}_B) \rho].    
\end{align}
Note that the $\sigma_{a \vert x}$ are sub-normalized quantum states, found with probability $p(a \vert x) = \mathrm{Tr}[\sigma_{a \vert x}]$, such that $ \sum_a \sigma_{a \vert x} = \rho_B = \mathrm{Tr}_A[\rho]. $ It is said that Alice can steer Bob, whenever the assemblage $\Vec{\sigma}$ does not admit an \ac{LHS} given by
\begin{align}
\sigma_{a \vert x} = \sum_{a'} p(a \vert x, a') \sigma_{a'},   
\end{align}
where the $ p(a \vert x, a') $ are post-processing probabilities and the $ \sigma_{a'}$ are sub-normalized states such that $ \mathrm{Tr}[\sum_{a'}  \sigma_{a'}] = 1$. It is straightforward to verify that in order for Alice to steer Bob, they need to share an entangled state $\rho$ and that Alice needs to perform incompatible measurements $\mathcal{M}$. In particular, it is known that the problem of deciding steerability and measurement incompatibility are equivalent, in the sense that there is a one-to-one mapping between the problems \cite{PhysRevLett.115.230402} and any incompatible assemblage $\mathcal{M}$ can lead to a steerable assemblage $\Vec{\sigma}$ if $\rho$ is chosen appropriately~\cite{Quintino2014JM,Uola2014JM}. \\
\indent Now, the authors in \cite{PhysRevA.109.062203} certified in such a scenario that Alice not only has access to more than one measurement (uses incompatible measurements) but that she can use genuinely $n$ measurements in the sense of measurement simulability according to Eq.~\eqref{Eq:SimuablityDef}. The idea is here, that a constraint in the available measurement assemblages $\mathcal{M}$, will result in a constraint in the possible steering assemblages $\Vec{\sigma}$ that Alice can prepare for Bob. However, as $n$-simulability is hard to characterize directly, the authors used the fact that $n$-simulability implies $n$-copy joint measurability. With that, the problem reduces to a convex problem which can be solved numerically, to certify the minimal number of measurements Alice must have access to, in order to prepare $\Vec{\sigma}$. \\
\indent From our results here, specifically from Result~\ref{result3-formal} and Result~\ref{result4-formal}, it follows that the approximation of $n$-simulability by $n$-copy joint measurability can be refined by considering $n$-wise compatibility instead. This not only promises better bounds on critical visibility, but also maintains the appealing factor of being efficiently computable, since $n$-wise compatibility can be checked via an \ac{SDP}. In particular, it is now obvious that using $n$-wise compatibility as an approximation of $n$-simulability leads to the best possible bounds achievable with convex optimization tools. Let us exemplify this in the case of the three noisy Pauli measurements. The authors in \cite{PhysRevA.109.062203} find that the three Pauli measurements have to be considered as genuinely $3$ measurements in the context of simulability in the semi-device independent scenario for visibilities larger than $\tfrac{\sqrt{3}}{2}$, which is simply the critical visibility for the $2$-copy joint measurability. Notably, the authors in \cite{PhysRevA.109.062203} improve this value by considering only \emph{separable} parent \acp{POVM} in the context of $n$-copy joint measurability, leading to a visibility of $\sqrt{\tfrac{2}{3}} \approx 0.8165$. 
Due to our Result~\ref{result3-formal} and Result~\ref{result4-formal} we further improve this value to a critical visibility of $\tfrac{\sqrt{2}+1}{3} \approx 0.8047$ by using $2$-wise compatibility as a refined approximation. \\
\indent It is noteworthy that due to the strong relation between measurement incompatibility and steering, our results can also directly be applied to study generalizations of steering. In particular\addMTQ{,} it holds
\begin{align}
\mathrm{Tr}_A\left[(M_{a \vert x} \otimes \mathds{1}_B) \lvert \Phi^{+} \rangle \langle \Phi^{+} \rvert\right] = \dfrac{M_{a \vert x}^{T}}{d},    
\end{align}
where $\lvert \Phi^{+} \rangle = \tfrac{1}{\sqrt{2}}\sum_{i=0}^{d-1} \lvert ii \rangle$ is the maximally entangled state and the transposition is with respect to the Schmidt-basis defined by $\lvert \Phi^{+} \rangle$. This implies directly that properties of $ \mathcal{M} $ can be translated into properties of $ \Vec{\sigma} $. 

\subsection{Implications for the device- and theory-independent certification of number of measurements}

Our work also has a close relationship to the work \cite{Tendick2025QuantumCorrelations} and provides an additional implication of the results therein. In the work, it is used that the number of quantum measurements one has access to, can also be certified device-independently if one considers the notion of $n$-wise compatibility, as first discussed in \cite{PhysRevLett.123.180401}. Even more, the number of measurements can also be certified in a \emph{theory-independent} way, i.e., by not using the quantum formalism specifically, but by only relying on the no-signaling principle. More precisely, let us consider a Bell experiment, in which two (distant and non-communicating) parties perform measurements indexed by $x$ for Alice and $y$ for Bob. The outcomes of the measurements are denoted by $a$ for Alice's measurements and similarly by $b$ for Bob's measurements. Treating their measurement devices as black-boxes (i.e., working in the device-independent framework), Alice and Bob make no assumptions about their measurement devices or the shared state and only make statements based on observations, i.e., on the basis of the input-output statistics $\mathbf{P} = \left\lbrace p(ab \vert xy) \right\rbrace_{a,b,x,y}$. In particular, in \cite{Tendick2025QuantumCorrelations}, the considered statistics $\mathbf{P}$ are only restricted by the no-signaling principle, i.e., $\mathbf{P}$ cannot be used for super-luminal communication. Formally, it has to hold
\begin{align}
\sum_a p(ab \vert xy) = p(b\vert xy) = p(b \vert y), \ \ \sum_b p(ab \vert xy) = p(a\vert xy) = p(a \vert x),
\end{align}
for all $a,b,x,y$. The point of the work \cite{Tendick2025QuantumCorrelations} is to show that there exist for any $n \geq 2$, $n$-wise incompatible quantum measurements, which can be certified on the basis of the input-output statistics $\mathbf{P} $ even when only the no-signaling principle is assumed (and not the formalism of quantum theory). This follows fundamentally from the fact that the violation of a Bell inequality requires incompatibility in any no-signaling theory \cite{PhysRevLett.103.230402}. This implies that for any finite $n \geq 2$, there exist genuinely $n$-wise incompatible measurements and there exist for any $n$ an experiment that cannot be reproduced with less than $n$ measurements in any no-signaling theory. \\
\indent As we saw in this work, the notion of $n$-wise compatibility is not the unique choice to define what it means to perform \emph{genuinely} $n$ measurements. However, our results imply that the results in \cite{Tendick2025QuantumCorrelations} directly hold true for all notions that are weaker than $n$-wise compatibility. In particular, it means that also from the perspective of $n$-simulability there exist quantum correlations obtained from $n$ quantum measurements that cannot be reproduced using arbitrary no-signaling correlations of less than $n$ measurements in the simulability framework. \\
\indent It is an interesting open question whether the result in \cite{Tendick2025QuantumCorrelations} can also be extended in any meaningful way to $n$-copy joint measurability. While being operationally less close, it would give an interesting generalization. The main question that needs to be answered to tackle this problem is how $n$-copy joint measurability can be certified device-independently. That is, which constraints on correlations $\mathbf{P} = \left\lbrace p(ab \vert xy) \right\rbrace_{a,b,x,y}$ arise from $n$-copy jointly measurable measurements on, say, Alice's side? Once this question is clarified, it is likely that it can also be translated into a theory-independent question (just relying on the no-signaling principle).

\section{Discussion}
\label{Sec:Discussion}

We have shown that the question of how many measurements an assemblage $\mathcal{M}$ genuinely constitutes is ambiguous and has been answered in prior literature using different ways to define $n$ genuine measurements. In particular, there are three main approaches to the question: measurement simulability, incompatibility structures, and multi-copy joint measurability. All concepts aim to generalize the notion of \emph{standard joint measurability}. We discussed the operational differences between the concepts and showed that they can be ordered in a hierarchy of strict set-inclusions, see Eq.~\eqref{Eq:MainResult} and Fig.~\ref{Fig:Circuits}. %
Our main result is achieved through a series of results (see Result~\ref{result1-nontechnical} to
Result~\ref{result4-nontechnnical}) relating the different subsets of measurement assemblages that have previously been defined.

Our results therefore provide a clear picture of the previously scrambled up mix of concepts and are meant to guide further research on the generalizations of measurement incompatibility. Our results can be summarized as follows. We have first shown that probabilistic pre-processing provides advantages over deterministic pre-processings for the notion of measurement simulability. Second, we combine computational methods with the $\epsilon$-grid approach based on error estimation arguments to show that the set of $n$-simulable measurement assemblages is generally non-convex. Third, we have shown that the convex hull of the set of $n$-simulable assemblages coincides with the set of $n$-wise compatible assemblages. Finally, we have shown that $n$-wise compatibility implies $n$-copy joint measurability and derived a universal lower bound on the the visibility of $n$-copy joint measurability on pure input states for any assemblage $\mathcal{M}$. \\
\indent One important message of our work is that the choice of which generalization of joint measurability is considered for a certain application depends highly on the context. The notions of $n$-copy joint measurability and $n$-simulability have clear operational meanings, which is somewhat lacking for $n$-wise compatibility, which has a clear geometric interpretation. This can be made evident through our reoccurring example of the three Pauli measurements and the role of freely accessible randomness. Let us consider the decomposition in Eq.~\eqref{Eq:DefConvHull}. This expression is operationally problematic, as it probabilistically mixes more than $n$ measurements. That is, to implement the simulation of an assemblage $\mathcal{M}$ according to the RHS of Eq.~\eqref{Eq:DefConvHull} one would, in practice, be forced to implement more than $n$ measurements. Concretely, let us consider the the three noisy Pauli measurements at the point of critical visibility $\eta = \tfrac{\sqrt{2}+1}{3}$, such that they are still $2$-wise compatible. According to the decomposition found in \cite{PhysRevLett.123.180401}, these noisy measurements can be implemented by probabilistically mixing (with probability $\tfrac{1}{3}$) assemblages where two out of three Pauli measurements are noisy (with visibility $\eta_2 = \tfrac{1}{\sqrt{2}}$, making them jointly measurable) while the remaining measurement is a noise-free Pauli measurement. That is, one needs to use noise-free Pauli measurements to simulate noisy Pauli measurements in the same direction, which is operationally difficult while being geometrically unproblematic. \\
\indent On the other hand, the notion of $n$-simulability and randomness assisted $n$-simulability defined in Eq.~\eqref{Eq:SimuablityDef} and  Eq.~\eqref{Eq:SimuablityDefRandomness} do not suffer from any such problems, as they are operationally well-defined. However, their geometric structure remains unclear and the lack of convexity makes it hard to study the set of $n$-simulable assemblages. Hence, it might often be a good idea to approximate these sets by their convexification, depending on the specific application in mind. \\
\indent Our work paves the way for various future research directions. First, our work shines light on the necessity to have a better analytical (or even efficient numerical) characterization of measurement simulability. This would help to strengthen our Result~\ref{result2-nontechnical} and would likely also give insights for the comparison of $n$-simulability and randomness assisted $n$-simulability. Furthermore, it would be interesting to see whether other, operationally relevant, generalizations of measurement incompatibility exist and how they would fit into the hierarchy in Eq.~\eqref{Eq:MainResult}. Note that it is straightforward to invent some new notions of incompatibility by combining some notions discussed here. For instance, while $n$-copy incompatibility is defined using a single parent \ac{POVM}, there is mathematically speaking no problem to combine it with the framework of measurement simulability or compatibility structures. In particular, $n$-copy incompatibility is also the only generalization of measurement incompatibility discussed in this manuscript, for which currently there is no method to certify the said incompatibility in a device-independent way. Finally, it would be interesting to see whether any of the generalizations of measurement incompatibility is directly linked to the performance of certain quantum information processing protocols and whether our results hold also beyond quantum theory, in so-called generalized probabilistic theories \cite{Pl_vala_2023}.

\begin{acknowledgments}
We thank Jessica Bavaresco, Jef Pauwels, Lucas Porto, Roope Uola, and Isadora Veeren for helpful discussions. LT acknowledges funding from the ANR through the JCJC grant LINKS (ANR-23-CE47-0003).
\\[1em]
\end{acknowledgments}

\section*{Code availability}

Codes that accompany this work are available on online repositories. See \cite{code-repo-R1} for Result \ref{result1-formal} and \cite{code-repo-R2} for Result \ref{result2-formal}.

\section*{Author contributions} 

Lucas Tendick took the lead in the project, developing main ideas and proofs, and wrote the first version of the manuscript. Costantino Budroni proposed the approach taken in order to derive Result $2$ and took the lead in obtaining it. Marco Túlio Quintino took the lead in deriving Result $5$ and in particular identified that it is only valid for pure input states. All authors contributed to discussions of the ideas and proofs and independently checked the results. All authors also contributed to revising the manuscript and preparing the final version. \\

\noindent No artificial intelligence has been used in the process of obtaining the presented results or in the writing process.

\begin{acronym}[CGLMP]\itemsep 1\baselineskip
\acro{AGF}{average gate fidelity}
\acro{AMA}{associated measurement assemblage}

\acro{BOG}{binned outcome generation}

\acro{CGLMP}{Collins-Gisin-Linden-Massar-Popescu}
\acro{CHSH}{Clauser-Horne-Shimony-Holt}
\acro{CP}{completely positive}
\acro{CPT}{completely positive and trace preserving}
\acro{CPTP}{completely positive and trace preserving}
\acro{CS}{compressed sensing} 

\acro{DFE}{direct fidelity estimation} 
\acro{DM}{dark matter}

\acro{GST}{gate set tomography}
\acro{GUE}{Gaussian unitary ensemble}

\acro{HOG}{heavy outcome generation}

\acro{JM}{jointly measurable}

\acro{LHS}{local hidden-state model}
\acro{LHV}{local hidden-variable model}
\acro{LOCC}{local operations and classical communication}

\acro{MBL}{many-body localization}
\acro{ML}{machine learning}
\acro{MLE}{maximum likelihood estimation}
\acro{MPO}{matrix product operator}
\acro{MPS}{matrix product state}
\acro{MUB}{mutually unbiased bases} 
\acro{MW}{micro wave}

\acro{NISQ}{noisy and intermediate scale quantum}

\acro{POVM}{positive operator valued measure}
\acro{PVM}{projector-valued measure}

\acro{QAOA}{quantum approximate optimization algorithm}
\acro{QML}{quantum machine learning}
\acro{QMT}{measurement tomography}
\acro{QPT}{quantum process tomography}
\acro{QRT}{quantum resource theory}
\acroplural{QRT}[QRTs]{Quantum resource theories}

\acro{RDM}{reduced density matrix}

\acro{SDP}{semidefinite program}
\acro{SFE}{shadow fidelity estimation}
\acro{SIC}{symmetric, informationally complete}
\acro{SM}{Supplemental Material}
\acro{SPAM}{state preparation and measurement}

\acro{RB}{randomized benchmarking}
\acro{rf}{radio frequency}

\acro{TT}{tensor train}
\acro{TV}{total variation}

\acro{UI}{uninformative}

\acro{VQA}{variational quantum algorithm}

\acro{VQE}{variational quantum eigensolver}

\acro{WMA}{weighted measurement assemblage}

\acro{XEB}{cross-entropy benchmarking}

\end{acronym}

\bibliographystyle{./myapsrev4-2}
\bibliography{bibliography_2.bib}

\end{document}